\documentclass[11pt]{article}
\usepackage{amsmath,amssymb,amsthm}
\usepackage{paralist} 
\usepackage{times}
\usepackage{epsfig,epsf,psfrag,array,verbatim,mathtools}
\usepackage{url,comment}
\usepackage[usenames,dvipsnames,svgnames,table]{xcolor}
\usepackage{wrapfig}

\usepackage[boxed,linesnumbered,vlined]{algorithm2e}
\SetKwInput{Input}{\hspace{-4mm}Input}

\usepackage{tikz}
\usetikzlibrary{arrows,patterns}
\usetikzlibrary{matrix,calc,decorations.pathreplacing}
\usepackage{xspace}
\newcommand{\zell}[2][c]{ \begin{tabular}[#1]{@{}c@{}}#2\end{tabular}}

\theoremstyle{plain}
\newtheorem{definition}{Definition}
\newtheorem{theorem}{Theorem}
\newtheorem{lemma}{Lemma}
\newtheorem{remark}{Remark}
\newtheorem{proposition}{Proposition}
\newtheorem{observation}{Observation}
\newtheorem{corollary}{Corollary}

\newtheorem*{open}{Open Question}
\newtheorem*{question.}{Question}
\newtheorem*{sub-property}{Substructure Property}

\newcommand{\ceil}[1]{\left\lceil #1 \right\rceil}
\newcommand{\floor}[1]{\left\lfloor #1 \right\rfloor}

\newcommand{\DeG}{\textsc{Deg}}

\newcommand{\addlayer}{\textsc{AddLayer}\xspace}
\newcommand{\maxdn}{\textsc{MaxNDeg}\xspace}
\newcommand{\mindn}{\textsc{MinNDeg}\xspace}
\newcommand{\maxdon}{\textsc{MaxNDeg$^-$}\xspace}
\renewcommand{\deg}{\textit{deg}}

\newcommand{\oldH}{H_{old}}
\newcommand{\oldL}{L_{old}}
\newcommand{\newH}{H_{new}}
\newcommand{\newL}{L_{new}}
\newcommand{\longsigma}{\sigma=(d_\ell^{n_\ell},\cdots,d_1^{n_1})}

\newcommand{\maxop}{\text{\raise.2em\hbox{$~\underset{\textup{max}}{\circledast}~$}}}
\newcommand{\minop}{\text{\raise.2em\hbox{$~\underset{\textup{min}}{\circledast}~$}}}
\newcommand{\sumop}{\text{\raise.2em\hbox{$~\underset{\textup{sum}}{\circledast}~$}}}

\newcommand{\ccon}{\textsc{CCon}\xspace}
\newcommand{\cgen}{\textsc{CGen}\xspace}
\newcommand{\ocon}{\textsc{OCon}\xspace}
\newcommand{\ogen}{\textsc{OGen}\xspace}
\newcommand{\ogenl}{\textsc{OGenL}\xspace}
\newcommand{\ogenu}{\textsc{OGenU}\xspace}

\newcommand{\leader}{\textit{leader}}
\renewcommand{\deg}{\textit{deg}}

\newcolumntype{q}{>{\centering\arraybackslash}p{8em}}
\newcolumntype{e}{>{\centering\arraybackslash}p{10em}}
\newcolumntype{r}{>{\centering\arraybackslash}p{16em}}
\newcolumntype{t}{>{\centering\arraybackslash}p{25em}}
\newcolumntype{y}{>{\centering\arraybackslash}p{28em}}
\definecolor{mycolor}{rgb}{0.09,0.08,0.97}

\title{Graph Realizations: Maximum and Minimum Degree in Vertex Neighborhoods}

\author{
Amotz Bar-Noy%
\thanks{City University of New York (CUNY), USA. E-mail: amotz@sci.brooklyn.cuny.edu},
Keerti Choudhary%
\thanks{Tel Aviv University, Israel. E-mail: keerti.choudhary@cs.tau.ac.il},
David Peleg%
\thanks{Weizmann Institute of Science, Israel. E-mail: david.peleg@weizmann.ac.il},
Dror Rawitz%
\thanks{Bar Ilan University, Israel. E-mail: dror.rawitz@biu.ac.il}
}
\date{}
\begin{document}
\maketitle

\begin{abstract}

The classical problem of \emph{degree sequence realizability} asks whether or not a given sequence of $n$ positive integers 
is equal to the degree sequence of some $n$-vertex undirected simple graph. While the realizability problem of degree sequences has been well studied for different classes of graphs, there has been relatively little work concerning the realizability of other types of information \emph{profiles}, such as the vertex neighborhood profiles.

In this paper, we initiate the study of {\em neighborhood degree} profiles, wherein, our focus is on the natural problem of realizing \emph{maximum} and \emph{minimum} neighborhood degrees.  More specifically, we ask the following question: \emph{``Given a sequence $D$ of $n$ non-negative integers $0\leq d_1\leq \cdots \leq d_n$, does there exist a simple graph with vertices $v_1,\ldots, v_n$ such that for every $1\le i \le n$, the maximum (resp. minimum) degree in the neighborhood of $v_i$ is exactly $d_i$?"}


We provide in this work various results for both maximum as well as minimum neighborhood degree for general $n$ vertex graphs. 
Our results are first of its kind that studies extremal neighborhood degree profiles. For maximum neighborhood degree profiles, we provide a {\em complete realizability criteria}.
In comparison, we observe that the minimum neighborhood profiles are not so well-behaved, for these our  necessary and sufficient conditions for realizability {\em differ by a factor of at most two}.


\end{abstract}

\section{Introduction}

In many application domains involving networks, it is common to view vertex degrees as a central parameter, providing useful information concerning the relative significance (and in certain cases, centrality) of each vertex with respect to the rest of the network, and consequently useful for understanding the network's basic properties. Given an $n$-vertex graph $G$ with adjacency matrix $Adj(G)$, its {\em degree sequence} is a sequence consisting of its vertex degrees,
$$\DeG(G) = (d_1,\ldots,d_n).$$
Given a graph $G$ or its adjacency matrix, it is easy to extract the degree sequence.
An interesting {\em dual} problem, sometimes referred to as the {\em realization} problem, concerns a situation where given a sequence of nonnegative integers $D$, we are asked whether there exists a graph whose degree sequence conforms to $D$. A sequence for which there exists a realization is called a {\em graphic} sequence. Erd\"os and Gallai~\cite{EG60} gave a necessary and sufficient condition for deciding whether a given sequence of integers is graphic (also implying an $O(n)$ decision algorithm). Havel and Hakimi~\cite{hakimi62,havel55} gave a recursive algorithm that given a sequences of integers computes in $O(m)$ time a realizing $m$-edge graph, or proves that the sequence is not graphic.

Over the years, various extensions of the degree realization problem were studied as well, cf. \cite{AT94,BCPR19-range-isaac,WK73}, concerning different characterizations of degree-profiles. The motivation underlying the current paper is rooted in the observation that realization questions of a similar nature pose themselves naturally in a large variety of {\em other} application contexts, where given {\em some} type of information profile specifying the desired vertex properties (be it concerning degrees, distances, centrality, or any other property of significance), it can be asked whether there exists a graph conforming to the specified profile. Broadly speaking, this type of investigation may arise, and find potential applications, both in scientific contexts, where the information profile reflects measurement results obtained from some natural network of unknown structure, and the goal is to obtain a model that may explain these measurements, and in engineering contexts, where the information profile represents a specification with some desired properties, and the goal is to find an implementation in the form of a network conforming to that specification.

This basic observation motivates a vast research direction, which was little studied over the last five decades. In this paper we make a step towards a systematic study of one specific type of information profiles, concerning {\em neighborhood degree} profiles. Such profiles are of theoretical interest in context of social networks (where degrees often reflect influence and centrality, and consequently neighboring degrees reflect ``closeness to power'').
Neighborhood degrees were considered before in ~\cite{BarrusD:18}, where the profile associated with each vertex $i$ is the {\em list} of degrees of all vertices in $i$'s neighborhood. In contrast, we focus here on ``single parameter'' profiles, where the information associated with each vertex relates to a single degree in its neighborhood.
Two first natural problems in this direction concern the {\em maximum} and {\em minimum} degrees in the vertex neighborhoods. For each vertex $i$, let $d'_i$ (respectively, $d''_i$) denote the maximum (resp., minimum) vertex degree in $i$'s neighborhood. Then $\maxdn(G)=(d'_1,\ldots,d'_n)$ (resp., $\mindn(G)=(d''_1,\ldots,d''_n)$) is the maximum (resp., minimum) neighborhood degree profile of $G$.
The same realizability questions asked above for degree sequences can be posed for neighborhood degree profiles as well. This brings us to the following central question of our work:

\begin{question.}
Can we efficiently compute for a given sequence $D=(d_1,\ldots,d_n)$ of nonnegative integers an $n$-vertex graph $G$ (if exists) such that the {\em maximum} (resp. {\em minimum}) degree in the neighborhood of $i$-$th$ vertex in $G$ is exactly equal to $d_i$? Moreover, is there a closed-form characterization for all $n$-length realizable sequences?
\end{question.}



\paragraph*{Our Contributions}

For simplicity, we represent the input vector $D$ alternatively in a more compact format as $\sigma=(d_\ell^{n_\ell}, \cdots, d_1^{n_1}),$ where $n_i$'s are positive integers with $\sum_{i=1}^\ell n_i = n$; here the specification requires that $G$ contains exactly $n_i$ vertices whose minimum (resp. maximum) degree in  
neighborhood is $d_i$.  We may assume that $d_\ell>d_{\ell-1}>\cdots>d_1\geq 1$
(noting that vertices with max/min degree zero are necessarily singletons and can be handled separately).

\vspace{2mm}
\noindent
\textit{\textbf{(a) Minimum Neighborhood degree:}}
In Section \ref{section:minNdeg} we show the following necessary and sufficient conditions for $\longsigma$ to be $\mindn$ realizable.

The necessary condition is that for each $i\in [1,\ell]$, 
\begin{align}
d_i   & \leq n_1 + n_2 + \ldots + n_i - 1 ~, ~~~~\text{and}
\tag{NC1} \label{NC1} \\[1mm]
d_\ell & \leq\Big\lfloor\frac{n_1d_1}{d_1+1}\Big\rfloor
+ \Big\lfloor\frac{n_2d_2}{d_2+1}\Big\rfloor 
+ \ldots + \Big\lfloor\frac{n_\ell d_\ell}{d_\ell+1}\Big\rfloor
\tag{NC2} \label{NC2}
~.
\end{align}

The sufficient condition is that for each $i\in [1,\ell]$,
\begin{align}
d_i & \leq \Big\lfloor\frac{n_1d_1}{d_1+1}\Big\rfloor
+ \Big\lfloor\frac{n_2d_2}{d_2+1}\Big\rfloor 
+ \ldots + \Big\lfloor\frac{n_id_i}{d_i+1}\Big\rfloor
~.
\tag{SC} \label{SC}
\end{align}

\begin{remark}
For any sequence $\longsigma$ satisfying the first necessary condition
\eqref{NC1}, the sequence
$\sigma^\gamma=(d_\ell^{\lceil\gamma n_\ell\rceil},\ldots,d_1^{\lceil\gamma n_1\rceil})$,
where $\gamma=(d_1+1)/d_1$ satisfies the sufficient condition \eqref{SC}, thus our necessary and sufficient conditions differ by a factor of at most 2 in the $n_i$'s.
\end{remark}

\begin{remark}
For $\ell$ bounded by $3$, we show that $\longsigma$ is $\mindn$-realizable if and only if 
along with \eqref{NC1} and \eqref{NC2} following is satisfied:

\begin{equation}
\text{Either $d_2\leq \Big\lfloor \frac{n_1d_1}{d_1+1}\Big\rfloor+\Big\lfloor \frac{n_2d_2}{d_2+1}\Big\rfloor$, or $d_3+1\leq n_1+n_2+n_3-\Big(1+\ceil{\frac{d_2-n_2}{d_1}}\Big)$}
\tag{NC3}\label{NC3}
\end{equation}
\end{remark}

We leave it as an open question to resolve the problem in general.

\begin{open}
Does there exist a closed-form characterization
for realizing $\mindn$ profiles for general graphs?
\end{open}


\vspace{1mm}
\noindent
\textit{\textbf{(b) Maximum Neighborhood degree:}}
We perform an extensive study of maximum neighborhood degree profiles.
\begin{enumerate}
\item In Section \ref{section:maxNdeg}, we obtain the necessary and sufficient conditions for $\longsigma$ to be $\maxdn$ realizable.

For general graphs we obtain the following characterization.
$$d_\ell \leq n_\ell-1, \text{ and } d_1\geq 2 \text{ or }n_1\text{ is even}$$

We also study the version of the problem in which the realization is required to be connected. Our characterization is as follows.
$$d_\ell \leq n_\ell-1, \text{ and } d_1\geq 2 \text{ or }\sigma=(1^2)~.$$

\item Further, we consider the open neighborhoods, wherein a vertex is not counted in its own neighborhood. These are more involved and are discussed in Section~\ref{section:maxNdegOpen}.
Our results for open neighborhood are summarised in Table~\ref{table:open_max_deg}.

\begin{table}[!ht]
\begin{center}
\def\arraystretch{1.1}
\begin{tabular}{|e| t|}
\hline
	\bf \zell {Graph} & \bf \zell{Complete characterisation}
\\ \hline	
		\zell{Connected Graphs} 
	& \zell{$d_\ell \leq \min\{n_\ell,n-1\}$ \\ $d_1\geq 2$  or $\sigma=(d^d,1^1)$ or $\sigma=(1^2)$\\[1mm]
			$\sigma\neq (d_\ell^{d_\ell+1},2^1)$} 
\\ \hline
	\zell{General graphs} 
	& \zell{$\sigma$ can be split\footnotemark~into two profiles 
				$\sigma_1$ and $\sigma_2$ such that\\ 
				(i) $\sigma_1$ has a {\em connected} $\maxdn$-open realization, and \\
				(ii) $\sigma_2=(1^{2\alpha})$ or $\sigma_2=(d^d, 1^{2\alpha+1})$, 
				for integers $d\geq 2, \alpha\geq 0$.}
\\ \hline		
\end{tabular}
\caption{Max-neighbouring-degree realizability for open neighborhood.\vspace{-6mm}}
\label{table:open_max_deg}
\end{center}
\end{table} 
\footnotetext{
A profile $\sigma=(d_\ell^{n_\ell},\cdots,d_1^{n_1})$ is said to be 
split into two profiles 
$\sigma_1=(d_\ell^{p_\ell},\cdots,d_1^{p_1})$
and $\sigma_2=(d_\ell^{q_\ell},\cdots,d_1^{q_1})$ if 
$n_i=p_i+q_i$ for each $i\in[1,\ell]$.
}

\item {\bf Enumerating realizable maximum neighborhood degree profiles:}\\
The simplicity of above characterizations enables us to enumerate and count the number of realizable profiles. This gives a way to sample uniformly a random \maxdn realizable profile. In contrast, counting and sampling are open problems for the traditional degree sequence realizability problem.
In Appendix, we show that the number of realizable profiles of length $n$ is $\displaystyle\lceil(2^{n-1}+(-1)^n)/3\rceil$ for general graphs and $2^{n-3}$ for connected graphs.
In comparison, the total number of non-increasing sequences of length $n$
on the numbers $1,\ldots,n-1$ is $\Theta(4^n/\sqrt{n})$.  
\end{enumerate}

In Section~\ref{section:discussion}, we discuss the apparent difference in difficulty between \maxdn and \mindn profiles and propose a possible explanation.

\paragraph*{Further Related Work}
Many works have addressed related questions such as finding all the (non-isomorphic) graphs that realize a given degree sequence, counting all the (non-isomorphic) realizing graphs of a given degree sequence, sampling a random realization for a given degree sequence as uniformly as possible, or determining the conditions under which a given degree sequence defines a unique realizing graph (a.k.a. the \emph{graph reconstruction} problem), cf.~\cite{choudum86,EG60,hakimi62,havel55,K57,O70,SH91,TT08,U60,W99}. Other works such as~\cite{BD11,Cloteaux16,MV02} studied interesting applications in the context of social networks.

To the best of our knowledge, the \maxdn and \mindn realization problems have not been explored so far. There are only two related problems that we are aware of. The first is the \emph{shotgun assembly} problem~\cite{MR15}, where the characteristic associated with the vertex $i$ is some description of its neighborhood up to radius $r$. The second is the {\em neighborhood degree lists} problem~\cite{BarrusD:18}, where the characteristic associated with the vertex $i$ is the list of degrees of all vertices in $i$'s neighborhood. We point out that in contrast to these studies, our \maxdn and \mindn problem applies to a more restricted profile (with a single number characterizing each vertex), and the techniques involves are totally different from those of~\cite{BarrusD:18,MR15}. Several other realization problems are surveyed in~\cite{BCPR18survey,BCPR19happiness}.


\section{Preliminaries}
\label{section:prelim}
Let $H$ be an undirected graph.  We use $V(H)$ and $E(H)$ to respectively denote the vertex set and the edge set of graph $H$.  For a vertex $x \in V(H)$, let $\deg_H(x)$ denote the degree of $x$ in $H$.  Let $N_H[x] = \{x\} \cup \{y~|~(x,y) \in E(H)\}$ be the (closed) neighborhood of $x$ in $H$. For a set $W \subseteq V(H)$, we denote by $N_H(W)$, the set of all the vertices lying outside set $W$ that are adjacent to some vertex in $W$, that is, $N_H(W) = (\bigcup_{w \in W} N[w]) \setminus W$. Given a vertex $v$ in $H$, the minimum (resp. maximum) degree in the neighborhood of $v$, namely $\mindn_H(v)$ (resp. $\maxdn_H(v)$), is defined to be the maximum over the degrees of all the vertices in the neighborhood of $v$.
Given a set of vertices $A$ in a graph $H$, we denote by $H[A]$ the subgraph of $H$ induced by the vertices of $A$. For a set $A$ and a vertex $x\in V(H)$, we denote by $A\cup x$ and $A\setminus x$, respectively, the sets $A\cup \{x\}$ and $A\setminus \{x\}$. When the graph is clear from context, for simplicity, we omit the subscripts $H$ in all our notations. Finally, given two integers $i \leq j$, we define $[i,j] = \{i,i+1,\ldots,j\}$.

\begin{figure}[!ht]
\centering
\begin{tikzpicture}[scale=0.65]
\begin{scope}[every node/.style={circle,draw,fill=black!5!White}]
\node (v1) at (0,0) {$3$};
\node (v2) at (-2,0) {$2$};
\node (v3) at (-1,1.5) {$2$};
\node (v4) at (2,0) {$2$};
\node (v5) at (4,0) {$1$};
\end{scope}
\draw (v1) -- (v2);
\draw (v1) -- (v3);
\draw (v1) -- (v4);
\draw (v4) -- (v5);
\draw (v2) -- (v3);
\end{tikzpicture}
\caption{A $\maxdn$ realization of $(3^4,2^1)$ and a $\mindn$ realization of $(2^3,1^2)$.}
\label{fig:example1}
\end{figure}
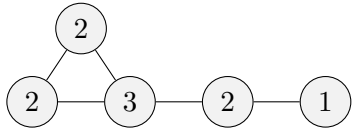

A profile $\sigma=(d_\ell^{n_\ell}, \cdots, d_1^{n_1})$ satisfying $d_\ell>d_{\ell-1}>\cdots>d_1> 0$ is said to be $\mindn$ realizable (resp. $\maxdn$ realizable) if there exists a graph $G$ on $n=n_1+\cdots+n_\ell$ vertices that for each $i \in [1,\ell]$ contains exactly $n_i$ vertices whose $\mindn$ (resp.~$\maxdn$) is $d_i$. Equivalently, $|\{v \in V(G) : \mindn(v) = d_i\}| = n_i$ (resp. $|\{v \in V(G) : \maxdn(v) = d_i\}| = n_i$).
The figure depicts a $\maxdn$ realization of $(3^4,2^1)$ and a $\mindn$ realization of $(2^3,1^2)$. (The numbers in the vertices represent their degrees.)
Note that in the open neighborhoods model, the corresponding $\maxdn$ and $\mindn$ profiles become $(3^3,2^2)$ and $(2^4,1^1)$, respectively.

\section{Realizing minimum neighborhood degree profiles}
\label{section:minNdeg}

\subsection{Leaders and followers}

Let $G=(V,E)$ be any graph. For any vertex $v\in V$, we define $\leader(v)$ to be a vertex in $N[v]$ of minimum degree, if there are more than one choices we pick the leader arbitrarily. In other words, $\leader(v) = \arg\min\{ \deg(w) \mid w\in N[v]\}$. Next let $\sigma=(d_\ell^{n_\ell} \cdots d_1^{n_1})$ be the min-degree sequence of $G$. We define $V_i$ to be set of those vertices in $G$ whose minimum-degree in the closed neighborhood is exactly $d_i$, so $|V_i|=n_i$. Also, let $L_i$ be set of those vertices in $G$ who are leader of at least one vertex in $V_i$, equivalently, $L_i=\{\leader(v)~|~v\in V_i\}$, and denote by $L=\cup_{i=1}^\ell L_i$ the set of all the leaders in $G$. Observe that the sets $V_1,\ldots,V_\ell$ forms a partition of the vertex-set of $G$. 

A vertex $v$ in $G$ is said to a {\em follower}, if $\leader(v)\neq v$. Let $F_i=\{v\in V_i~|~ v\neq \leader(v)\}$
be the set of all the followers in $V_i$. Finally we define $R=V\setminus L$ to be the set of all the non-leaders, and 
$F=\cup_{i=1}^\ell F_i$ to be the set of all the followers.

\begin{figure}[!ht]
\centering
\begin{footnotesize}
\begin{tikzpicture}[scale=1.2]
\begin{scope}[every node/.style={circle,draw,minimum size=0.66cm,fill=black!5!White}]
\node (v1) at (0,0) {$v_1$};
\node (v2) at (1,0) {$v_2$};
\node (v3) at (2,0) {$v_3$};
\node (v4) at (3,1) {$v_4$};
\node (v5) at (3,-1) {$v_5$};
\node (v6) at (4,0) {$v_6$};
\end{scope}
\draw (v1) -- (v2);
\draw (v2) -- (v3);
\draw (v3) -- (v4);
\draw (v3) -- (v5);
\draw (v3) -- (v6);
\draw (v4) -- (v5);
\draw (v4) -- (v6);
\draw (v5) -- (v6);
\end{tikzpicture}
\end{footnotesize}
\caption{The unique $\mindn$-realization of the sequence $\sigma=(3^3 2^1 1^2 )$. Observe that $\textit{min-deg}(v_1) = \textit{min-deg}(v_2) = \deg(v_1) = 1$, $\textit{min-deg}(v_3) = \deg(v_2) = 2$, and $\textit{min-deg}(v_i) = 3$, for $i \in \{4,5,6\}$. Since $\leader(v_2)=v_1$ and $\leader(v_3)=v_2$, here $v_2$ is a leader as well as a follower.}
\label{fig:example}
\end{figure}
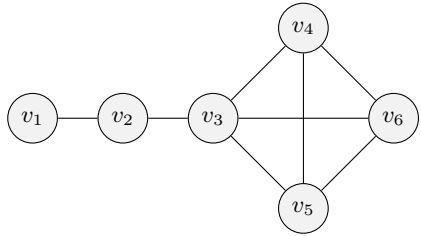

We point here that there exist realizable sequences $\sigma$ for which any graph $G$ realizing $\sigma$ and any leader function over $G$, the sets $L$ and $F$ have non-empty intersection. For example, consider the sequence $\sigma=(1^2 2^1 3^3)$ in Figure~\ref{fig:example}. It can be easily checked that $\sigma$ has only
one realizing graph, and in this graph, the leader-set and the follower-set are not disjoint.

We classify the sequences that admit disjoint leader and follower sets as follows.
\begin{definition}
A sequence $\sigma=(d_\ell^{n_\ell} \cdots d_1^{n_1})$ is said to admit a \textit{Disjoint Leader-Follower} (DLF) $\mindn$-realization if there exists a graph $G$ realizing $\sigma$ and a $\leader$ function under which the sets $L$ and $F$ are mutually disjoint, that is, $L\cap F=\emptyset$. 
\end{definition}

\subsection{Realizing uniform sequences}
\begin{lemma}
For a sequence $\sigma=(d_\ell^{n_\ell} \cdots d_1^{n_1})$ to be $\mindn$-realizable it is necessary that  $d_1+1\leq n_1$.
\label{lemma:x_1}
\end{lemma}
\begin{proof}
Suppose $\sigma$ is $\mindn$-realizable by a graph $G$, then there exists at least one vertex, say $w$, of degree exactly $d_1$ in $G$. Now $|N[w]|=d_1+1$, this implies that the number of vertices $v$ in graph $G$ with $\textit{min-deg}(v)=d_1$ must be at least $d_1+1$. Thus $n_1\geq d_1+1$.
\end{proof}

\begin{lemma}\label{lemma:construction-single-case}
The sequence $\sigma=(d^{n})$, is $\mindn$-realizable if and only if $n\geq d+1$.
\end{lemma}
\begin{proof}
By Lemma~\ref{lemma:x_1}, if the sequence $\sigma=(d^{n})$ is realizable then $n$ must be at least $d+1$. 
To prove the converse, we give a realization for $\sigma$ assuming $n\geq d+1$. Let $q\geq 1$ and $r\in[0,d]$ be integers satisfying $n=(q)(d+1)-r$. Take a set $A$ of $q$ vertices, namely $L_i~(i\in [1,q])$, and another set $B$ of $dq$ vertices, namely $b_{ij}~(i\in [1,q],j\in[1,d])$. Connect each $L_i$ to vertices $b_{i1},\ldots,b_{id}$. So vertices in $A$ have degree exactly $d$ and vertices in $B$ have in their neighborhood a vertex of degree $d$. Next if $r>0$, then we merge $b_{1j}$ with $b_{2j}$, for $j\in[1,r]$, thereby reducing $r$ vertices in $B$. (Notice that $b_{1j}$ and $b_{2j}$ exists because $r>0$ only if $q\geq 2$). Thus $|A|+|R|=n$ and each vertex in $A$ still has degree exactly $d$.  So $|A|=\frac{n+r}{d+1}=\big\lceil \frac{n}{d+1}\big\rceil$ and $|R|=n-|A|=\big\lfloor \frac{nd}{d+1}\big\rfloor\geq d$. Finally, we add edges between each pair of vertices in $B$ to make it a clique of size at least $d$; this will imply that the vertices in set $B$ have degree at least $d$. It is easy to check that $\textit{min-deg}(v)$ for each $v\in A\cup B$ in our constructed graph is $d$.
\end{proof}

\begin{remark}
Henceforth, we will use $\textsc{graph}(n,d,A,B)$ to denote the function that returns the edges of the graph constructed by Lemma~\ref{lemma:construction-single-case} whenever $n\geq d+1$ and $|A|=\big\lceil \frac{n}{d+1}\big\rceil$, and $|R|=\big\lfloor \frac{nd}{d+1}\big\rfloor$.
\end{remark}

\subsection{Necessary and sufficient conditions for $\mindn$ profiles}

We start with the following theorem.

\begin{theorem}[Sufficient condition SC] \label{theorem:construction-easy-instance}
Any sequence $\sigma=(d_\ell^{n_\ell}\cdots d_1^{n_1})$ satisfying $d_i \leq \sum_{j=1}^i \Big\lfloor \frac{n_jd_j}{d_j+1}  \Big\rfloor$, for $i\in [1,\ell]$, is $\mindn$-realizable by a graph $G$ such that  $L\cap F= \emptyset$ with respect to some leader function defined over $G$.
\end{theorem}
\begin{proof}
We initialize $G$ to be an empty graph. Our algorithm proceeds in $\ell$ rounds. (See Algorithm \ref{Algorithm:min-deg} for a pseudo-code). In each round, we first add to $G$ a set $V_i$ of $n_i$ new vertices and partition $V_i$ into two sets $L_i$ and $R_i$ of sizes respectively $\big\lceil \frac{n_i}{d_i+1}\big\rceil$ and $\big\lfloor \frac{n_id_i}{d_i+1}\big\rfloor$. Now if $n_i> d_i+1$, then we solve this round independently by adding to $G$ all the edges returned by $\textsc{graph}(n_i,d_i,L_i,R_i)$. Notice that if $n_i\leq d_i+1$, then $L_i$ will contain only one vertex, say $a_i$. In such a case, we add edges between $a_i$ and all the vertices in set $R_i$. Also, we add edges between $a_i$ and any arbitrarily chosen $d_i+1-n_i$ vertices in $\cup_{j<i}R_j$. This is possible since $d_i+1-n_i=d_i - \big\lfloor \frac{n_id_i}{d_i+1}  \big\rfloor \leq \sum_{j=1}^{i-1} \big\lfloor \frac{n_jd_j}{d_j+1}  \big\rfloor= \sum_{j=1}^{i-1} |R_j|$. Finally, after the $\ell$ rounds are completed, we add edges between each pair of vertices in set $R = \cup_{i=1}^\ell R_i$ to make it a clique. 

\begin{algorithm}[!ht]
\Input{A sequence $\sigma=(d_\ell^{n_\ell} \cdots d_1^{n_1})$ satisfying 
$d_i \leq  \sum_{j=1}^i \lfloor \frac{n_jd_j}{d_j+1}\rfloor$, for $1\leq i\leq \ell$.}
\BlankLine
Initialize $G$ to be an empty graph.\\
\For{$i=1$ to $\ell$ }{
Add to $G$ a set $V_i$ of $n_i$ new vertices.\\
Partition $V_i$ in two sets $L_i$, $R_i$ such that $|L_i|= \big\lceil \frac{n_i}{d_i+1}\big\rceil$
and $|R_i|= \big\lfloor \frac{n_id_i}{d_i+1}\big\rfloor$.\\
\uIf{$(n_i> d_i+1$, \normalfont{or equivalently,} $|L_i|>1)$}
{Add to $G$ all the edges returned by $\textsc{graph}(n_i,d_i,L_i,R_i)$.}
\ElseIf{$(|L_i|=1)$}{ 
Let $a_i$ be the only vertex in $L_i$.\\
Connect $a_i$ to all vertices in $R_i$, and any arbitrary $d_i+1-n_i$ vertices in $\cup_{j<i}R_j$.}
}
Add edges between each pair of vertices in $R = \cup_{i=1}^\ell R_i$ to make it a clique.\\
Output G.
\caption{Computing a $\mindn$-realization for a given special $\sigma$.}
\label{Algorithm:min-deg}
\end{algorithm}

Let us now show bounds on the degree of vertices in sets $L_i$ and $R_i$.
\begin{enumerate}
\item Each vertex in $L_i$ has degree exactly $d_i$ :~
Recall we add edges to vertices in $L_i$ only in the $i^{th}$ iteration of for loop. If $n_i>d_i+1$, then by Lemma~\ref{lemma:construction-single-case}, the degree of each vertex in $L_i$ is exactly $d_i$. If $|L_i|=1$, or equivalently, $n_i\leq d_i+1$, then $|R_i|=n_i-|L_i|=n_i-1$, and so degree of vertex $a_i\in L_i$ is $(n_i-1)+(d_i+1-n_i)=d_i$.\\

\item Vertices in $R$ have degree at least $d_\ell$ :~
For any $i\in [1,\ell]$, if $n_i>d_i+1$, then by Lemma~\ref{lemma:construction-single-case}, $|R_i|=\big\lceil\frac{n_i d_i}{d_i+1}\big\rceil $, and even in the case $n_i\leq d_i+1$, we have $|R_i|=n_i-|L_i|=n_i-\big\lceil\frac{n_i }{d_i+1}\big\rceil =\big\lceil\frac{n_i d_i}{d_i+1}\big\rceil$. Thus $|R|=\sum_{i=1}^\ell |R_i| = \sum_{i=1}^\ell \big\lceil\frac{n_i d_i}{d_i+1}\big\rceil$ which  is bounded below by $d_i$. Since $|R|\geq d_\ell$, and each vertex in $R$ is adjacent to at least  one vertex in $\cup_i L_i$, the degree of vertices in $R$ is at least $d_\ell$.
\end{enumerate}

We next show that for any vertex $v\in V_i$, $\textit{min-deg}(v)=d_i$, where $i\in [1,\ell]$. If $v\in L_i$, then $\textit{min-deg}(v)=d_i$, since each vertex in $L_i$ has degree $d_i$, and is adjacent to only vertices in $R$ which have degree at least $d_\ell\geq d_i$. If $v\in R_i$, then also $\textit{min-deg}(v)=d_i$, since each vertex in $R_i$ is adjacent to at least one vertex in $L_i$, and $N[v]$ is contained in the set $R\cup (\cup_{j\geq i} L_j)$, whose vertices have degree at least $d_i$.

The leader function over $V$ is as follows. For each $v\in \cup_{i=1}^\ell L_i$, we set $\leader(v)=v$, and for each $v\in R_i$, we set $\leader(v)$ to any arbitrary neighbour of $v$ in $L_i$. Since each vertex in $L=\cup_{i=1}^\ell L_i = \{\leader(v)~|~v\in V\}$ is a leader of itself, the set $L$ of leader and the set $F$ of followers must be mutually disjoint.
\end{proof}

We now provide a lower bound on the size of the leader set $L_i$.

\begin{lemma}
For each $i\in[1,\ell]$, we have $|L_i|\geq \displaystyle \Big\lceil \frac{n_i}{d_i+1}\Big\rceil$.
\label{lemma:L_i}
\end{lemma}
\begin{proof}
Consider any vertex $a \in L_i$. Since $|N(a)|=d_i+1$, vertex $a$ can serve as leader for at most $d_i+1$ vertices. This shows that $|L_i|\geq \frac{n_i}{d_i+1}$. The claim follows from the fact that $|L_i|$ is an integer.
%
\end{proof}

\begin{theorem}[Necessary condition] \label{theorem:necess}
For any $\mindn$-realizable sequence $\sigma=(d_\ell^{n_\ell} \cdots d_1^{n_1})$, we have
\begin{description}
\item{{\bf (NC1)}} {~~}
  $d_i\leq \big(\sum_{j=1}^{i}n_j\big) - 1$, for $i\in[1,\ell]$
\item{{\bf (NC2)}} {~~}
  $d_\ell \leq \displaystyle \sum_{i=1}^{\ell}\Big\lfloor \frac{n_id_i}{d_i+1}  \Big\rfloor$.
\end{description}
\end{theorem}
\begin{proof}
Let $G$ be a realization for $\sigma$. Let $w$ be any vertex in $G$ such that $\deg(w)=d_i$. Then $w$ as well as all the neighbours of $w$ must be contained in $\cup_{j=1}^i V_j$, therefore, $d_i+1 =|N[w]|\leq |\cup_{j=1}^iV_j| = \sum_{j=1}^{i}n_j$, implying condition (NC1).

To prove condition (NC2), suppose $w$ is a vertex in $G$ such that $\textit{min-deg}(w)=d_\ell$. Then $N[w]$ cannot contain vertices of degree less than $d_\ell$, so $N[w]\cap L_i=\emptyset$, for each $i<\ell$. Therefore, $|N[w]|\leq n - \sum_{i=1}^{\ell-1}|L_i|$. Also $\deg(w)$ must be at least $d_\ell$. We thus get, 
$$d_{\ell}+1\leq |N[w]| \leq n - \sum_{i=1}^{\ell-1}|L_i| = n_\ell + \sum_{i=1}^{\ell-1}(n_i-|L_i|)\leq n_\ell + \sum_{i=1}^\ell  \Big\lfloor \frac{n_id_i}{d_i+1}\Big\rfloor, $$
where the last inequality follows from Lemma~\ref{lemma:L_i}. 

If $n_\ell\leq d_\ell$, then $n_\ell-1 =\big\lfloor \frac{n_\ell d_\ell}{d_\ell+1}\big\rfloor$, and so $d_\ell \leq  \sum_{i=1}^{\ell}\big\lfloor \frac{n_id_i}{d_i+1}  \big\rfloor$. If $n_\ell\geq d_\ell +1$, then $\frac{n_\ell d_\ell}{d_\ell+1}\geq d_\ell$ which implies $d_\ell\leq\big\lfloor \frac{n_\ell d_\ell}{d_\ell+1}\big\rfloor$ since $d_\ell$ is integral.
\end{proof}

As a corollary of the above results, the following is immediate.
\begin{corollary}
The sequence $\sigma=(d_2^{n_2}d_1^{n_1})$ is $\mindn$-realizable if and only if $d_1\leq \big\lfloor \frac{n_1d_1}{d_1+1}\big\rfloor$ and $d_2\leq \big\lfloor \frac{n_1d_1}{d_1+1}\big\rfloor+\big\lfloor \frac{n_2d_2}{d_2+1}\big\rfloor$.
\end{corollary}
\begin{proof}
Suppose $\sigma=(d_2^{n_2}d_1^{n_1})$ is realizable. Then Theorem~\ref{theorem:necess} implies (i) $n_1\geq d_1+1$ which implies $d_1\leq \big\lfloor \frac{n_1d_1}{d_1+1}\big\rfloor$, and (ii) $d_\ell=d_2\leq \big\lfloor \frac{n_1d_1}{d_1+1}\big\rfloor+\big\lfloor \frac{n_2d_2}{d_2+1}\big\rfloor$. The converse follows from Theorem~\ref{theorem:construction-easy-instance}.
\end{proof}

For a sequence $\longsigma$,
let $\gamma=(d_1+1)/d_1$.
As $\lfloor\frac{\gamma n_1d_1}{d_1+1}\rfloor + \ldots +
\lfloor\frac{\gamma n_id_i}{d_i+1}\rfloor \geq n_1+\cdots+n_i\geq d_i$,
we also have the following.

\begin{corollary}
For any sequence $\longsigma$ satisfying the first necessary condition
\eqref{NC1},
the sequence $\sigma^\gamma=(d_\ell^{\lceil\gamma n_\ell\rceil},\ldots,d_1^{\lceil\gamma n_1\rceil})$ 
satisfies the sufficient condition \eqref{SC}.
\end{corollary}

\subsection{\mindn realization of tri-sequences}
\label{section:minNdeg-3}

We here  consider the scenario when a sequence has only three distinct degrees.
Specifically, we provide a complete characterization of sequences $\sigma=(d_3^{n_3}d_2^{n_2}d_1^{n_1})$. 

\begin{theorem}
The necessary and sufficient conditions for $\mindn$-realizability of the sequence $\longsigma$ when $\ell=3$ is
\begin{enumerate}
\item $d_1+1\leq n_1$, 
\item $d_2+1\leq n_1+n_2$, 
\item  $d_3\leq \big\lfloor \frac{n_1d_1}{d_1+1}\big\rfloor+\big\lfloor \frac{n_2d_2}{d_2+1}\big\rfloor+\big\lfloor \frac{n_3d_3}{d_3+1}\big\rfloor$, and
\item  either $d_2\leq \big\lfloor \frac{n_1d_1}{d_1+1}\big\rfloor+\big\lfloor \frac{n_2d_2}{d_2+1}\big\rfloor$, or $d_3+1\leq n_1+n_2+n_3-\big(1+\ceil{\frac{d_2-n_2}{d_1}}\big)$.
\end{enumerate}

\end{theorem}
\begin{proof}
Suppose $\sigma=(d_3^{n_3}d_2^{n_2}d_1^{n_1})$ is realizable, then by Theorem~\ref{theorem:necess}, it follows that the first three conditions stated above are necessary. 

To prove that all four conditions are necessary, we are left to show that if $d_2\gneq \big\lfloor \frac{n_1d_1}{d_1+1}\big\rfloor+\big\lfloor \frac{n_2d_2}{d_2+1}\big\rfloor$, then $d_3+1\leq n_1+n_2+n_3-\big(1+\ceil{\frac{d_2-n_2}{d_1}}\big)$.
We consider a graph $G$ that realizes $\sigma$. Let $V_1,V_2,V_3$ be the partition of $V(G)$ as defined in Section~\ref{section:minNdeg}. Consider a vertex $w\in V_2$. Observe that $leader(w)$ must lie in $V_1$, because if $L_2\cap V_2$ is non-empty, then Lemma~\ref{lemma:disjoint_L_i_V_i} implies $d_2\leq \big\lfloor \frac{n_1d_1}{d_1+1}\big\rfloor+\big\lfloor \frac{n_2d_2}{d_2+1}\big\rfloor$. We first show that $|L_1|\geq \ceil{\frac{d_2-n_2}{d_1}}$. The set $N(w)\cap V_1$ has size at least $d_2-n_2$. Each vertex $x\in L_1$ can serve as a leader of at most $d_1$ vertices in open-neighborhood of $w$. Indeed, if $x\in N(w)$ then it can not count $w$ (lying outside $N(w)$), and if $x\notin N(w)$ then it can not count itself (again lying outside $N(w)$). Thus to cover the set $N(w)\cap V_1$ at least $\ceil{\frac{d_2-n_2}{d_1}}$ leaders are required, thereby, showing $|L_1|\geq \ceil{\frac{d_2-n_2}{d_1}}$. Now consider a vertex $y\in V_3$, note that $N[y]$
excludes $w$ (as degree of $w$ is $d_2$), as well as $L_1$ (as vertices in $L_1$ have degree $d_1$). Therefore, we obtain the following relation.
$$d_3+1=|~N[y]~|\leq |V_1\setminus L_1| + |V_2\setminus w| + |V_3|\leq  n_1+n_2+n_3-\Big(1+\ceil{\frac{d_2-n_2}{d_1}}\Big)$$

We now prove the sufficiency claims. If $d_2\leq \big\lfloor \frac{n_1d_1}{d_1+1}\big\rfloor+\big\lfloor \frac{n_2d_2}{d_2+1}\big\rfloor$, then the conditions 1-4 are sufficient by Theorem~\ref{theorem:construction-easy-instance}. So let us focus on the scenario when $d_2\gneq \big\lfloor \frac{n_1d_1}{d_1+1}\big\rfloor+\big\lfloor \frac{n_2d_2}{d_2+1}\big\rfloor$. Let $N=n_1+n_2+n_3-\big(1+\ceil{\frac{d_2-n_2}{d_1}}\big)$. The vertex-set of our realized graph $G=(V,E)$ will be a union of three disjoint sets $L_1,L_2=\{w\}$, and $Z$ of size respectively $\ceil{\frac{d_2-n_2}{d_1}}$, $1$, and $N$. Initially, the edge-set $E$ is an empty-set. Between vertex pairs in $Z$, we add edges so that the induced graph $G[Z]$ is identical to $\textsc{graph}(N,d_3,\ceil{\frac{N}{d_3+1}},\floor{\frac{Nd_3}{d_3+1}})$. This step is possible since $d_3+1\leq N$, and ensures that $\mindn_{G[Z]}(z)=d_3$, for $z\in Z$. Let $L_3$ denote the set of those vertices in $Z$ whose degree is equal to $d_3$. We connect $w$ to arbitrary $N-n_3=n_2+(n_1-|L_1\cup L_2|)$ vertices in $Z\setminus L_3$, and any arbitrary $\alpha:=d_2-(n_1+n_2-|L_1\cup L_2|)$ vertices in $L_1$. Since $\deg_G(w)=d_2$, this step ensures that $\mindn$ of exactly $n_2$ vertices in $Z$ decreases to $d_2$.
Let $Y$ be a subset of arbitrary $(n_1-|L_1\cup L_2|)$ neighbours of $w$ in $Z$.
Finally, we connect each $x\in L_1\cap N[w]$ to arbitrary $d_1-1$ vertices in $Y$, and each $x'\in L_1\setminus N[w]$ to arbitrary $d_1$ vertices in $Y$, so as to ensure each vertex in $Y$ is adjacent to at least one leader in $L_1$. Since vertices in $L_1$ have degree $d_1$, this ensures $\mindn_{G}(x)=d_1$, for each $x\in \{w\}\cup Y\cup L_1$. This completes the construction of $G$. 
\end{proof}

Looking at the complexity of the above characterization, we leave it as an open question to solve the problem in general.

\begin{open}
Does there exist a polynomial-time algorithm, or a closed-form characterization,
for realizing $\mindn$ profiles for general graphs?
\end{open}

\subsection{Complete characterization for sequences admitting disjoint leader-follower sets}
We conclude by providing a complete characterization for special class of $\mindn$-sequences that admit a disjoint leader-follower sets.

\begin{lemma} \label{lemma:disjoint_L_i_V_i}
Let $G$ be a graph and $\sigma(G)=(n_\ell^{d_\ell}\ldots d_1^{n_1})$. For any leader function defined over $G$ and for any $i\in[1,\ell]$, if $L_i\cap V_i$ is non-empty then $d_i \leq \sum_{j=1}^i \Big\lfloor \frac{n_jd_j}{d_j+1}  \Big\rfloor$.
\end{lemma}

\begin{proof}
Let $w$ be any vertex lying in $L_i\cap V_i$, so $\textit{min-deg}(w)=\deg(w)=d_i$. Recall for each $j<i$, vertices in the set $L_j$ have degree strictly less than $d_i$. Since $N[w]$ cannot contain vertices of degree less than $d_i$, thus for each $j<i$, $N[w]\cap L_j=\emptyset$. Also vertices in $V_{i+1}\cup\ldots \cup V_\ell$ cannot be adjacent to any vertex in $\{w\}\cup\big(\cup_{j=1}^{i-1}L_j\big)$, therefore, $N[w]$ as well as $\cup _{j=1}^{i-1}L_j$ are contained in union $\cup_{j=1}^i V_j$. We thus get,
\begin{align*}
d_{i}+1= |N[w]| &\leq \Big|\bigcup_{j=1}^i V_j\Big| - \Big| \bigcup _{j=1}^{i-1}L_j\Big|
= n_i + \sum_{j=1}^{i-1}(n_i-|L_j|)\leq n_i + \sum_{j=1}^{i-1}  \Big\lfloor \frac{n_jd_j}{d_j+1}\Big\rfloor,
\end{align*}
where the last inequality follows from Lemma~\ref{lemma:L_i}. If $n_i\leq d_i$, then $n_i-1=n_i - \big\lceil \frac{n_i}{d_i+1}\big\rceil =\big\lfloor \frac{n_id_i}{d_i+1}\big\rfloor$, and so $d_i \leq  \sum_{j=1}^{i}\big\lfloor \frac{n_jd_j}{d_j+1}  \big\rfloor$. If $n_i\geq d_i+1$, then the bound trivially holds since $\frac{n_i d_i}{d_i+1}\geq d_i$ which from the fact that~$d_i$ is integral implies $d_i\leq \big\lfloor \frac{n_id_i}{d_i+1}\big\rfloor$.
\end{proof}


\begin{theorem} \label{theorem:disjoint_L_i_V_i}
A sequence $\sigma=(n_\ell^{d_\ell}\ldots d_1^{n_1})$ is $\mindn$-realizable by a graph $G$ having disjoint leader-set $(L)$ and follower-set $(F)$ with respect to some leader function, if and only if, for each $i\in[1,\ell]$, $d_i  \leq \sum_{j=1}^i \big\lfloor \frac{n_jd_j}{d_j+1}  \big\rfloor$.
\end{theorem}

\begin{proof}
Let us suppose there exists a leader function over $G$ for which $L\cap F=\emptyset$, then for each $i\in [1,\ell]$, $L_i\subseteq V_i$. This is because if for some $i$, there exists $w\in L_i\setminus V_i$, then $\deg(w)=d_i\neq \textit{min-deg}(d_i)$, which implies that $w$ is a leader as well as a follower. Since $L_i\subseteq V_i$, by Lemma~\ref{lemma:disjoint_L_i_V_i}, $d_i  \leq \sum_{j=1}^i \big\lfloor \frac{n_jd_j}{d_j+1}  \big\rfloor$, for each $i\in [1,\ell]$. The converse claim follows from Theorem~\ref{theorem:construction-easy-instance}.
\end{proof}

\section{Realizing maximum neighborhood degree profiles}
\label{section:maxNdeg}

In this section, we provide a complete characterization of $\maxdn$  profiles.
For simplicity, we first discuss the uniform scenario of 
 $\sigma=(d^k)$.
Observe that a star graph $K_{1,d}$ is $\maxdn$ realization of the profile $(d^{d+1})$. 
We show in the following lemma that, by identifying together vertices in different copies of $K_{1,d}$, it is always possible to realize the profile $(d^{k})$, whenever $k\geq d+1$.

\begin{lemma}~\label{lemma:simple-case}
For any positive integers $d$ and $k$, the profile $\sigma=(d^{k})$ is $\maxdn$ realizable whenever $k\geq d+1$. Moreover, we can always compute in $O(k)$ time a connected realization that has an independent set, say $S$, of size $d$ such that all vertices in $S$ have degree at most $2$, and at least two vertices in $S$ have degree $1$.
\end{lemma}

\begin{proof}
Let $\alpha$ be the smallest integer such that $k \leq 2 + \alpha(d-1)$. We first construct a caterpillar\footnote{A caterpillar is a tree in which all the vertices are within distance one of a central path.} $T$ as follows.
Take a path $P=(s_0,s_1,\ldots,s_\alpha,s_{\alpha+1})$ of length $\alpha+1$. Connect each internal vertex $s_i$ (here $i\in [1,\alpha]$) with a set of $d-2$ new vertices, so that the degree of $s_i$ is $d$. (See Figure~\ref{fig:caterpillar}). Note that the $\maxdn$ of each vertex $v\in T$ is $d$. 

\begin{figure}[!ht]
\centering
\begin{footnotesize}
\begin{tikzpicture}[scale=0.65]
\begin{scope}[every node/.style={circle,draw,minimum size=0.66cm,fill=black!5!White}]
\node (v0) at (0.5,0) {$s_0$};
\node (v1) at (4,0) {$s_1$};
\node (v2) at (9,0) {$s_2$};
\node (v3) at (14,0) {$s_3$};
\node (v4) at (17.5,0) {$s_4$};
\node
(u11) at (2.6,2.5) {$s^1_1$};
\node
(u12) at (4,2.5) {$s_1^2$};
\node (u13) at (5.4,2.5) {$s_1^3$};
\node
(u21) at (7.6,2.5) {$s^1_2$};
\node
(u22) at (9,2.5) {$s^2_2$};
\node (u23) at (10.4,2.5) {$s^3_2$};
\node (u31) at (12.5,2.5) {$s^1_3$};
\node (u32) at (14,2.5) {$s^2_3$};
\node (u33) at (15.5,2.5) {$s^3_3$};
\end{scope}
\draw (v0) -- (v1);
\draw (v1) -- (u11);
\draw (v1) -- (u12);
\draw (v1) -- (u13);
\draw (v1) -- (v2);
\draw (v2) -- (u21);
\draw (v2) -- (u22);
\draw (v2) -- (u23);
\draw (v2) -- (v3);
\draw (v3) -- (u31);
\draw (v3) -- (u32);
\draw (v3) -- (u33);
\draw (v3) -- (v4);
\draw[dotted,thick,<->] (u11) to [bend left=60] (u21);
\draw[dotted,thick,<->] (u12) to [bend left=60] (u22);
\end{tikzpicture}
\end{footnotesize}
\caption{A caterpillar for $d = 5$ and $\alpha = 3$.  If $k = 12$, then $r
  = 2$, and we merge
  (i) $s^1_1$ and $s^1_2$, and
  (ii) $s^2_1$ and $s^2_2$.
}
\label{fig:caterpillar}
\end{figure}
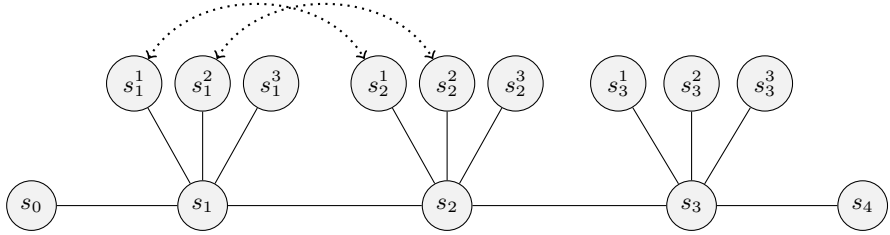

Now if $k = 2+\alpha(d-1)$, then $T$ serves as our required realizing graph. If $k < 2 + \alpha(d-1)$, then $\alpha \geq 2$ since $k \geq d+1$. The tree $T$ is ``almost'' a realizing graph for the profile, except that
it has too many vertices. Let $r = 2 + \alpha(d-1) - k$ denote the number of excess vertices in $T$ that need to be \emph{removed}. The $r$ vertices can be removed as follows.  Take any two distinct internal vertices $s_i$ and $s_j$ on $P$, and let $s^1_i,\ldots, s^{d-2}_i$ and $s^1_j,\ldots, s^{d-2}_j$, respectively, denote the neighbors of $s_i$ and $s_j$ not lying on $P$.  Let $G$ be the graph obtained by merging vertices $s^\ell_i$ and
$s^\ell_j$ into a single vertex for $\ell \in [1,r]$.  (See Figure~\ref{fig:caterpillar}).  Since the number of vertices was decreased by $r$, $G$ now contains exactly $n$ vertices.  The degree of vertices $s_1,s_2,\ldots,s_\alpha$ remains $d$, and the degree of all other vertices is at most $2$, therefore $\maxdn(v)=d$ for each $v\in G$, so $G$ is a realization of the profile $\sigma$.

Finally, in the resultant graph $G$, the end points of $P$ (i.e. $s_0$ and $s_{\alpha+1}$) have degree $1$, and there are $d-2$ other vertices, namely $s^1_i, \ldots, s^{d-2}_i$ (or $s^1_j, \ldots, s^{d-2}_j$), that have degree bounded by $2$. Therefore we set $S$ to these $d$ vertices. It is easy to verify that $S$ is indeed an independent set.
\end{proof}


\subsection{An incremental procedure for computing \maxdn realizations}

We explain here our main building block, procedure \addlayer, that  will be useful in  incrementally building graph realizations in a decreasing order of maximum degrees. Given a partially computed connected graph $H$ and integers $d$ and $k$ satisfying $d\geq 2$ and $k\geq 1$, the procedure adds to $H$ a set $W$ of $k$ new vertices such that $\maxdn(w)=d$, for each $w\in W$. The reader may assume that $\maxdn(v)\geq d$, 
for each existing vertex $v \in V(H)$. The procedure takes in as an input a sufficiently large vertex list $L$ (of size $d-1$) that forms an independent set in $H$, and whose vertices have small degree (that is, at most $d-1$). Moreover, in order to accommodate its iterative use, each invocation of the procedure also generates and outputs a \emph{new} list, to be used in the further iterations.

\vspace{-2mm}
\subparagraph{Procedure \addlayer}
The input to procedure \addlayer($H,L,k,d$) is a connected graph $H$ and a list $L=(a_1,\ldots,a_{d-1})$ of vertices in $H$ whose degree is bounded above by $d-1$. The first step is to add to $H$ a set of $k$ new vertices $W =\{w_1,w_2,\ldots,w_k\}$.  Next, the new vertices are connected to the vertices of $L$ and to themselves so as to ensure that $\maxdn(w) = d$ for every $w\in W$.  Depending upon whether or not $k <d$, there are two separate cases. (Refer to Algorithm~\ref{algo:Add-layer} for pseudocode).

Let us first consider the case $k \leq d-1$. In this case we add edges from vertices in $W$ to a subset of vertices from $L$ such that those vertices in $L$ will have degree $d$ and therefore will imply $\maxdn(w) =  d$, for every $w \in W$. We initialize two variables, $count$ and $i$, respectively, to $k$ and $d-1$. The variable $count$ holds, at any instant of time, the number of vertices in $W$ that still need to be connected to vertices in $L$. While $count >0$, the procedure performs the following steps:~
\begin{inparaenum}[(i)]
\item compute $r=\min\{d-\deg(a_i),count\}$, the maximum number of vertices in $W$ that can be connected to vertex~$a_i$;
\item connect $a_{i}$ to following $r$ vertices in $W$: $w_{count-(r-1)},w_{count-(r-2)},\ldots,w_{count-1},w_{count}$; and
\item decrease $count$ by $r$, and $i$ by $1$.
\end{inparaenum}

When $count=0$, the vertices $a_i,a_{i+1},\ldots,a_{d-1}$ are connected to at least one vertex in $W$ (this implies $d-i\leq k$). It is also easy to verify that at this stage, $\deg(a_{d-1})=\deg(a_{d-2})=\cdots=\deg(a_{i+1})=d$, and $\deg(a_i)\leq d$. Since the input graph $H$ was connected, in the beginning of the execution $\deg(a_i) \ge 1$, and by connecting $a_i$ to at least one vertex in $W$, specifically to $w_1$, its degree is increased at least by one.  So at most $d-2$ edges need to be added to $a_i$ to ensure that its degree is exactly $d$.  The procedure performs the following operation for each $j\in [d-1,d-2,\ldots,2,1]$ (in the given order) until $\deg(a_i)=d$:~
\begin{inparaenum}[(i)]
\item if $j<i$ then add edge $(a_j,a_i)$ to $H$, and
\item if $j>i$ then add an edge between $a_i$ and an arbitrary
  neighbor of $a_j$ lying in $W$.
\end{inparaenum}
Since $\deg(a_i)=\deg(a_{i+1})=\cdots=\deg(a_{d-1})=d$, and $\deg(w)\le 2$ for every $w \in W$, it follows that $\maxdn(w) =  d$, for each $w \in W$. In the end, we set a \emph{new} list $L$ containing the first $d-2$ vertices in the sequence $(w_1,w_2,\ldots, w_k, a_{1}, a_2, \ldots, a_{i-1})$.  This is possible since $k+i-1\geq d-2$ due to the fact that $d-i\leq k$.  (Later on we bound the degrees of the vertices in the new list.)

\begin{algorithm}[!ht]
\SetProcNameSty{textsc}
\caption{\addlayer($H,L,k,d$)\label{algo:Add-layer}}
Let the list $L$ be $(a_1,a_2,\ldots,a_{d-1})$.\\
Add to $H$ a set $W=\{w_1,\ldots,w_k\}$ of $k$ new vertices.\\
\uCase{$(k< d)$}
{
Set $count=k$ and $i=d-1$.\\
\While{$(count\neq 0)$}
{
Let $r=\min\{d-\deg(a_i),count\}$.\\
Add edges $(a_{i},w_{count-t})$ to $H$ for $t\in [0,r-1]$.\\
Decrement $i$ by $1$ and $count$ by $r$.\\
}
\ForEach{$j\in [d-1,\ldots,2,1]$}{
If $\deg(a_i)=d$ then break the for loop.\\
If $(j<i)$ then add edge $(a_j,a_i)$ to $H$.\\
If $(j>i)$ then add an edge between $a_i$ and an arbitrary vertex in $N(a_j)\cap W$.\\
}
Set $L$ to be prefix of $(w_1,w_2,\ldots, w_k, a_{1}, a_2, \ldots, a_{i-1})$ of size $d-2$.\\
}
\Case{$(k\geq d)$}{
Use Lemma~\ref{lemma:simple-case} to compute over independent set $(W\cup \{a_1\})$ 
the graph, say $\bar H$, realizing the profile $(d^{k+1})$ such that $\deg_{\bar H}(a_1)=1$.\\ 
Add edges between $a_1$ and any arbitrary $d-\deg(a_{1})$ vertices in set $\{a_{2},a_{3},\ldots,a_{d-1}\}$.\\
Let $b_1,\ldots,b_{d-1}\in \bar H\setminus a_1$
be such that $1 = \deg_{\bar H}(b_1) \leq
\cdots\leq \deg_{\bar H}(b_{d-1})\leq 2$.\\
Set $L=(b_1,b_2,\ldots,b_{d-2})$.\\
}
Output $L$.\\
\end{algorithm}
\begin{algorithm}[!ht]
\caption{$\maxdn$ realization of $\sigma=(d_\ell^{n_\ell},\ldots, d_1^{n_1})$
\label{algo:MaxNDeg}}
\Input{A sequence $\longsigma$ satisfying $d_\ell\leq n_\ell-1$ and $d_1\geq 2$.}
\BlankLine
Initialize $G$ to be the graph obtained from Lemma~\ref{lemma:simple-case}
that realizes the profile $(d_\ell^{n_\ell})$.\\
Let $L_{\ell-1}$ be a valid list in $G$ of size $d_{\ell-1}-1$.\\
\For{$(i=\ell-1$ to $1)$}{
$L_{i-1}\gets$ \addlayer$(G, L_{i},n_i, d_i)$.\\
Truncate list $L_{i-1}$ to contain only the first $d_{i-1}-1(\leq d_i -2)$ vertices.\\
}
Output $G$.
\end{algorithm}

Now we consider the case $k \geq d$. The procedure uses Lemma~\ref{lemma:simple-case} to compute over the independent set $W \cup \{a_1\}$ a graph $\bar{H}$ realizing the profile $(d^{k+1})$ such that $\deg_{\bar H}(a_1) = 1$.  Notice that in the beginning of the execution, $\deg(a_1) \in [1,d-1]$, and it is increased by one by adding $\bar{H}$ over the set $W \cup \{a_1\}$. So now $\deg(a_1)\in [2,d]$.  To ensure $\deg(a_1) = d$, at most $d-2$ more edges need to be added to $a_1$.  Edges are added between $a_1$ and any arbitrary $d - \deg(a_{1})$ vertices in set $\{a_{2},a_{3},\ldots,a_{d-1}\}$.  This ensures that every $w \in W$ has $\maxdn(w) = d$. By Lemma~\ref{lemma:simple-case}, $\bar{H} \setminus \{a_1\}$ contains an independent set of $d-1$ vertices, say $b_1,\ldots,b_{d-1}$, such that $1 = \deg_{\bar H}(b_1) \leq \deg_{\bar H}(b_2)\leq \cdots\leq \deg_{\bar H}(b_{d-1})\leq 2$.  In the end, the procedure creates a \emph{new} list $L=(b_1,b_2,\ldots,b_{d-2})$.  

For sake of better understanding, in the rest of paper, we denote by $\oldH,\oldL$ and $\newH,\newL$ respectively the graph and the list before and after the execution of Procedure \addlayer. Observe that $V(\newH) = V(\oldH) \cup W$.

The following two lemmas follow from the description of algorithm.

\begin{lemma} \label{lemma:layer1}
Each $w\in W$ satisfies $\maxdn(w)=d$, and $N(w)\subseteq W\cup \oldL$.
\end{lemma}

\begin{lemma} \label{lemma:layer2}
Each $a\in \oldL\setminus \newL$ satisfies $\deg_{H_{new}}(a)\leq d$, and each $a\in \oldL\cap \newL$ satisfies $\deg_{H_{new}}(a)\leq \deg_{H_{old}}(a)+1$.
\end{lemma}

It is also easy to verify that the total execution time of Procedure \addlayer is $O(k+d)$.

\subparagraph{The Inheritance Property} 
Till now, we showed that given an independent list of $d-1$ vertices of degree at most $d-1$ in a graph $H$, we can add $k\geq 1$ vertices to $H$ such that the $\maxdn$ of these $k$ vertices is $d$.  In order to iteratively use this algorithm to add vertices of smaller $\maxdn$ values $(\lneq d)$ we require that the list $\newL$ computed by Procedure \addlayer should satisfy following three constraints:~
\begin{inparaenum}[(i)]
\item The size of $\newL$ should be $d-2$;
\item the vertices of $\newL$ should form an independent set; and most importantly,
\item the vertices in $\newL$ should have degree at most $d-2$.
\end{inparaenum}

In order to ensure these constraints on $\newL$, we further impose the constraint that the list $\oldL$ is a valid list; this is formally defined as below. 


\begin{definition}[{\small Valid List}]
\label{definition:valid}
A list $L=(a_1,a_2,\ldots,a_t)$ in a graph $G$ is said to be ``valid''
with respect to $G$ if the following two conditions hold:
\begin{inparaenum}[(i)]
\item for each $i\in [1,t]$, $\deg(a_i)\leq i$, and
\item the vertices of $L$ form an independent set in $G$.
\end{inparaenum}
\end{definition}

We next prove the inheritance property of our procedure.

\begin{lemma}[Inheritance property] \label{lemma:valid}
If the input list $\oldL$ in Procedure \addlayer is valid, then the output list $\newL$ is valid as well.
\end{lemma}

\begin{proof}
We first consider the case $k\leq d-1$. Let $i$ be the smallest index such that vertices $a_i,a_{i+1},\ldots,a_{d-1}$ are adjacent to some vertex of $W$ in $\newH$. (That is, $i$ is the index when Procedure \addlayer exits the while loop).  Recall that in the graph $\newH$, $w_1 \in W$ is a neighbor of $a_i$.  Also, to increase the degree of $a_i$ to $d$, we connect $a_i$ to some/all vertices in $a_{1},\ldots,a_{i-1}$, and some/all neighbors of $a_{i+1},\ldots,a_{d-1}$ lying in $W$.  Therefore the vertex set $W\cup \{a_1,\ldots,a_{i-1}\}$ is independent in $\newH$.  Also, its size at least $d-1$, as we showed that $k\geq d-i$. Since the list $\oldL=(a_1,a_2,\ldots,a_{d-1})$ is valid in the beginning of the execution of Procedure \addlayer, it follows that in $\oldH$, $\deg(a_j)\leq j$ for $j\in [1,d-1]$.  So by Lemma~\ref{lemma:layer2}, in $\newH$, (i) $\deg(a_j)\leq j+1$ for $j\in [1,i-1]$, (ii) $\deg(w_1)=1$, and (iii) the degree of each other vertex in $W\setminus w_1$ is at most $2$.  Consequently, $(w_1,\cdots,w_k)$ is a valid list of length at least $d-i\geq 1$. Since $\deg(a_j)\leq j+1$ for $j\in [1,i-1]$, the list $(w_1,\cdots,w_k,a_1,\ldots,a_{i-1})$ is valid and has length at least $d-1$.  Truncating this to length $d-2$ again gives us a valid list.

We now consider the case $k\ge d$. By Lemma~\ref{lemma:simple-case}, $H[W \cup \{a_1\}] = \bar{H}$ contains an independent set $\{b_1,b_2,\ldots,b_{d-1}\} \subseteq W$ such that $\deg(b_1)=1$ and $\deg(b_j)\leq 2$ for $j\in[2,d-1]$.  Therefore, $(b_1,b_2,\ldots,b_{d-2})$ is a valid list of length $d-2$ in $\newH$.
\end{proof}

The following proposition summarizes the above discussion.

\begin{proposition} 
For any integers $d\geq 2$, $k\geq 1$, and any connected graph $H$ containing a valid list $L$ of size $d-1$, procedure $\addlayer$ adds to $H$ in $O(k+d)$ time, a set $W$ of $k$ new vertices such that $\maxdn(w)=d$, for every $w\in W$. All the edges added to $H$ lie in $W\times (W\cup L)$. Moreover, $\deg_H(a)\leq d$, for every $a\in L$, and the updated graph remains connected and contains a new valid list of size $d-2$.
\label{proposition:valid}
\end{proposition}


\vspace{-2mm}
\subsection{The main algorithm}
We now present the main algorithm for computing the realizing graph
using Procedure \addlayer. 

Let $\sigma=(d_\ell^{n_\ell}, \cdots, d_1^{n_1})$ be any profile satisfying $d_\ell \leq n_\ell - 1$ and $d_1\geq 2$.  The construction of a connected graph realizing $\sigma$ is as follows (refer to Algorithm~\ref{algo:MaxNDeg} for pseudocode).  We first use Lemma~\ref{lemma:simple-case} to initialize $G$ to be the graph realizing the profile $(d_\ell^{n_\ell})$. Recall $G$ contains an independent set, say $W=\{w_1,w_2,\ldots,w_{d_\ell}\}$, satisfying the
condition that the degree of the first two vertices is one, and the degree of the remaining vertices is at most two. Set $L_{\ell-1} =(w_1,w_2,\ldots,w_{d_{\ell-1}-1})$ (notice that $d_{\ell-1}-1\leq d_\ell$).  It is easy to verify that this list is valid. Next, for each $i=\ell-1$ to $1$, perform the following steps:%
\begin{compactenum}[(i)]
\item Taking as input the valid list $L_{i}$ of size $d_{i}-1$, execute Procedure \addlayer$(G, L_{i},n_i, d_i)$ to add $n_i$ new vertices to $G$.  The procedure returns a valid list $L_{i-1}$ of size $d_i-2$.
\item Truncate the list $L_{i-1}$ to contain only the first $d_{i-1}-1(\leq d_i -2)$ vertices.  The truncated list remains valid since any prefix of a valid list is valid.
\end{compactenum}

\vspace{-2mm}
\subparagraph{Proof of Correctness}
Let $V_\ell$ denote the set of vertices in graph $G$ initialized in step 1, and for $i\in [1,\ell-1]$, let $V_i$ denote the set of $n_i$ new vertices added to graph $G$ in iteration $i$ of the for loop. Also for $i\in [1,\ell]$, let $G_i$ be the graph induced by vertices $V_i\cup \cdots \cup V_\ell$. The following lemma proves the correctness.

\begin{lemma} \label{lemma:mdcn-algo-correct}
For any $i\in [1,\ell]$, graph $G_i$ is a $\maxdn$ realization of profile  $(d_\ell^{n_\ell},\cdots,d_i^{n_i})$, and for any $j\in [i,\ell]$ and any $v\in V_j$, $\deg_{G_i}(v) \leq \maxdn_{G_i}(v) = d_j$.
\end{lemma}

\begin{proof}
We prove the claim by induction on the iterations of the for loop. The base case is for index $\ell$, and by
Lemma~\ref{lemma:simple-case} we have that $\deg_{G_\ell}(v) \leq \maxdn_{G_\ell}(v) = d_\ell$, for every $v \in V_\ell$. For the inductive step, we assume that the claim holds for $i+1$, and prove the claim for $i$. Consider any vertex $v$ in $G_i$.  We have two cases.
\begin{compactenum}
\item $v \in V_i~:$ In this case by Proposition~\ref{proposition:valid} we have that $\deg_{G_i}(v) \leq \maxdn_{G_i}(v) = d_i$.
\item $v \in V_j$, for $j > i~:$ We first show that for
  any vertex $w\in N_{G_i}[v]$, $\deg_{G_i}(w)\leq d_j$.  If $w\in V_i$,
  then we already showed $\deg_{G_i}(w) \leq d_i$. So let us consider
  the case $w\in V_{i+1}\cup\cdots\cup V_{\ell}$.  Now if $w \in L_i$
  participates in Procedure \addlayer$(G, L_{i},n_i, d_i)$, then by
  Proposition~\ref{proposition:valid}, in the updated graph
  $\deg_{G_i}(w)\leq d_i \lneq d_j$.  If $w\not\in L_i$, then the
  degree of $w$ is unaltered in the $i^{th}$ iteration, and thus
  $\deg_{G_i}(w)=\deg_{G_{i+1}}(w)\leq \maxdn_{G_{i+1}}(v) = d_j$ by the inductive
  hypothesis.  It follows that $\maxdn(v)$ remains unaltered due to
  iteration $i$, and thus $\maxdn_{G_i}(v)=\maxdn_{G_{i+1}}(v)=d_j$.
\end{compactenum}\vspace{-5.2mm}
\end{proof}

The execution time of the algorithm is $O\big(\sum_{i=1}^\ell (n_i+d_i)\big)$. This is also optimal.  Indeed, any connected graph realizing $\sigma$ must contain $\Omega(n_1+n_2+\cdots+n_\ell)$ edges as the degrees of all vertices must be non-zero. Also, the graph must contain at least one vertex of each of the degrees $d_1, d_2,\ldots,d_\ell$, and therefore must have $\Omega(d_1+d_2+\cdots+d_\ell)$ edges. In other words, any realizing graph must contain $\Omega\big(\sum_{i=1}^\ell (n_i+d_i)\big)$ edges, and thus the computation time must be at least $\Omega\big(\sum_{i=1}^\ell (n_i+d_i)\big)$. The following theorem is immediate from the above discussions.

\begin{theorem} \label{theorem:maxdeg_algo}
There exists an algorithm that given any profile $\sigma=(d_\ell^{n_\ell},\ldots, d_1^{n_1})$ satisfying $d_\ell\leq n_\ell-1$ and $d_1\geq 2$ computes in optimal time a connected $\maxdn$ realization of $\sigma$.
\end{theorem}

\vspace{-2mm}
\subsection{A complete characterization for $\maxdn$ realizable profiles}

The necessary conditions for $\maxdn$ realizability is as follows.

\begin{lemma} \label{lemma:suff_maxdeg}
A necessary condition for a profile $\sigma=(d_\ell^{n_\ell}, \cdots, d_1^{n_1})$ to be $\maxdn$ realizable is $d_\ell \leq n_\ell - 1$.
\end{lemma}

\begin{proof}
Suppose $\sigma$ is $\maxdn$ realizable by a graph $G$. Then $G$ must contain a vertex, say $w$, of degree $d_\ell$ in $G$.  Since $d_\ell$ is the maximum degree in $G$, the $\maxdn$ of all the $d_\ell+1$ vertices in $N[w]$ must be~$d_\ell$. Thus $n_\ell \geq d_\ell+1$.
\end{proof}

Consider a profile $\sigma=(d_\ell^{n_\ell}, \cdots, d_1^{n_1})$ realizable by a connected graph. If $d_1=1$, then the graph must contain a vertex, say $v$, of degree $1$, and the vertices in $N[v]$ must also have degree $1$. The only possibility for such a graph is a single edge graph on two vertices. Thus in this case $\sigma=(1^2)$. If $d_1\geq 2$, then by Lemma~\ref{lemma:suff_maxdeg}, for $\sigma$ to be realizable in this case we need that $n_\ell\geq d_\ell+1$. Also, by Theorem~\ref{theorem:maxdeg_algo}, under these two conditions $\sigma$ is always realizable. We thus have the following theorem.

\begin{theorem} \label{thm:char_closed_con}
For a profile $\sigma=(d_\ell^{n_\ell}, \cdots, d_1^{n_1})$ to be $\maxdn$ realizable by a connected graph the necessary and sufficient condition is that either
\begin{inparaenum}[(i)]
\item $n_\ell\geq d_\ell+1$ and $d_1\geq 2$, or
\item $\sigma=(1^2)$.
\end{inparaenum}
\end{theorem}

Now if $d_1=1$, then $n_1$ must be even, since the vertices $v$ with $\maxdn(v)=1$ must form a disjoint union of exactly $n_1/2$ edges.  So for general graphs we have the following theorem.

\begin{theorem} \label{thm:char_closed_gen}
For a profile $\sigma=(d_\ell^{n_\ell}, \cdots, d_1^{n_1})$ to be $\maxdn$ realizable by a general graph the necessary and sufficient conditions are that $d_\ell \geq n_\ell-1$, and either $n_1$ is even or $d_1\geq 2$.
\end{theorem}

\section{Realizing maximum open neighborhood-degree  profiles}
\label{section:maxNdegOpen}

We start by formally defining the realizable profiles for maximum degree in open neighborhood.

\begin{definition}[$\maxdon$ realizable profile]
A profile $\sigma=(d_\ell^{n_\ell}, \cdots, d_1^{n_1})$ is said to be $\maxdon$ realizable if there exists a graph
$G$ on $n=n_1+\cdots+n_\ell$ vertices that for each $i \in [1,\ell]$ contains exactly $n_i$ vertices whose $\maxdon$ is $d_i$. Equivalently, $|\{v \in V(G) : \maxdon(v) = d_i\}| = n_i$\footnote{For a vertex $v$ in $H$, the maximum degree in the open neighborhood ($N_H[v]\setminus v$) of vertex $v$, namely $\maxdon_H(v)$ is defined to be the maximum over the degrees of all the vertices present in the open neighborhood of~$v$.}.
\end{definition}

\begin{figure}[!ht]
\centering
\begin{tikzpicture}[scale=0.65]
\begin{scope}[every node/.style={circle,draw,fill=black!10!White}]
\node (v1) at (0,0) {$3$};
\node (v2) at (-2,0) {$2$};
\node (v3) at (-1,1.5) {$2$};
\node (v4) at (2,0) {$2$};
\node[fill=black!.5!White] (v5) at (4,0) {$1$};
\end{scope}
\node (labelnum) at (1,-1) {$(a)$};
\draw (v1) -- (v2);
\draw (v1) -- (v3);
\draw (v1) -- (v4);
\draw (v4) -- (v5);
\draw (v2) -- (v3);
\end{tikzpicture}
\quad\quad\quad\quad
\begin{tikzpicture}[scale=0.65]
\begin{scope}[every node/.style={circle,draw,fill=black!10!White}]
\node[fill=black!.5!White] (v1) at (0,0) {$3$};
\node (v2) at (-2,0) {$2$};
\node (v3) at (-1,1.5) {$2$};
\node (v4) at (2,0) {$2$};
\node[fill=black!.5!White] (v5) at (4,0) {$1$};
\end{scope}
\node (labelnum) at (1,-1) {$(b)$};
\draw (v1) -- (v2);
\draw (v1) -- (v3);
\draw (v1) -- (v4);
\draw (v4) -- (v5);
\draw (v2) -- (v3);
\end{tikzpicture}
\caption{A comparison of the $\maxdn$ realization of $(3^4,2^1)$ and a $\maxdon$ realization of $(3^3,2^2)$. 
}
\label{fig:comparison}
\end{figure}
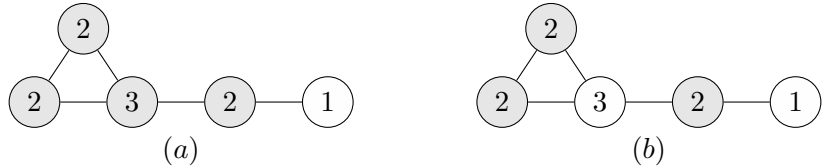

Observe that in the case of $\maxdon$ profiles, unfortunately, the nice {\em sub-structure property} (see Section~\ref{section:discussion}) does not always hold. For example, for the graph considered in Figure~\ref{fig:comparison}, the profile $\sigma=(3^3,2^2)$ is $\maxdon$ realizable, however, the subsequence $(3^3)$ is not $\maxdon$ realizable.

\subsection{Pseudo-valid List}
We begin by stating the following lemmas that are an extension of Lemma~\ref{lemma:simple-case} and Proposition~\ref{proposition:valid} presented in Section~\ref{section:maxNdeg} for \maxdn profiles. 

\begin{lemma}
For any positive integers $d$ and $k$, the profile $\sigma=(d^{k})$ is $\maxdon$ realizable whenever $k\geq d+2$. Moreover, we can always compute in $O(k)$ time a connected realization that contains an independent set having
\begin{inparaenum}[(i)]
\item two vertices of degree $1$, and
\item $d-2$ other vertices of degree at most $2$.
\end{inparaenum}
\label{lemma:simple-case-open}
\end{lemma}

\begin{proposition} 
For any integers $d\geq 2$, $k\geq 1$, and any connected graph $H$ containing a valid list $L$ of size $d-1$, procedure $\addlayer$ adds to $H$ in $O(k+d)$ time, a set $W$ of $k$ new vertices such that $\maxdon(w)=d$, for every $w\in W$. All the edges added to $H$ lie in $W\times (W\cup L)$. Moreover, $\deg_H(a)\leq d$, for every $a\in L$, and the updated graph remains connected and contains a new valid list of size $d-2$.
\label{proposition:valid-open}
\end{proposition}

It is important to note that though the Proposition~\ref{proposition:valid-open} holds for the open-neighborhoods it can not be directly used to incrementally compute the realizations. This is due to the reason that for the profiles $\sigma=(d_\ell^{d_\ell+1})$ unlike the scenario of $\maxdn$ realization, there is no $\maxdon$ realization that contains a valid list (See Lemma~\ref{lemma:d^d+1} for further details). 

This motivates us to define pseudo-valid lists.

\begin{definition}
A list $L=(a_1,a_2,\ldots,a_t)$ in a graph $H$ is said to be
``pseudo-valid'' with respect to $H$ if
\begin{inparaenum}[(i)]
\item for each $i\in [1,t]$, $\deg(a_i)=2$, and
\item the vertices of $L$ form an independent set.
\end{inparaenum}
\label{definition:pseudo-valid}
\end{definition}
Note that the only deviation that prevents $L$ from being a valid list is that $\deg(a_1)$ is 2 instead of 1.\\

We next state two lemmas that are crucial in obtaining $\maxdon$ realizations in the scenarios $n_\ell=d_\ell$ and $n_\ell=d_\ell+1$.

\begin{lemma}
\label{lemma:d^d_dbar^1}
For any integers $d > \bar d \geq 2$, the profile $\sigma=(d^{d},\bar
d^{1})$ is $\maxdon$ realizable. 
Moreover, in $O(d)$ time we can compute a connected realization  
that contains a valid list of size $d-1$.
\end{lemma}

\begin{proof}
The construction of $G$ is as follows.  Take a vertex $z$ and
connect it to $d-1$ other vertices $v_1,\ldots,v_{d-1}$.  Next
take another vertex $y$ and connect to $v_1, \ldots, v_{\bar d-1}$ 
(recall $2\leq \bar d<d$). Also connect $z$ to $y$.  
In the resulting graph $G$, $\deg(z) = d$, $\deg(y) = \bar d$, 
and $\deg(v_i) \leq 2$ for $i \in [1,d-1]$.  Also, $v_{d-1}$
is not adjacent to $y$ as $\bar d<d$, thus $\deg(v_{d-1})=1$.
Therefore, $\maxdon(z) = \bar d$, $\maxdon(y) = d$, and
$\maxdon(v_i)=d$, for $i\in [1,d-1]$. It is also easy to verify that 
$(v_{d-1},\ldots,v_{\bar d-1},\ldots,v_2,v_{1})$ is a valid list in $G$.
\end{proof}

\begin{lemma}
For any integer $d \geq 2$, the profile $\sigma=(d^{d+1})$ is
$\maxdon$ realizable.  Moreover, a connected realization that 
contains an independent set having $d-1$ vertices of degree $2$ 
can be compute in $O(d)$ time. 
However, none of the graphs realizing $\sigma$ can contain 
a vertex of degree $1$.
\label{lemma:d^d+1}
\end{lemma}

\begin{proof}
The construction of graph $G$ realizing $\sigma$ is very similar to
the previous lemma. Take two vertex-sets, namely, $U=\{u_1,u_2\}$ and
$W=\{w_1,\ldots,w_{d-1}\}$. Add to $G$ the edge $(u_1,u_2)$, and for
each $i\in [1,d-1]$, add to $G$ the edges $(u_1,w_i)$ and $(u_2,w_i)$.
This ensures that $\deg(u_1)=\deg(u_2)=d$ and $\deg(w_i)=2$ for $i\in
[1,d-1]$.  So $G$ contains $d+1$ vertices with $\maxdon$ equal to $d$.
Also, $W$ is an independent set of size $d-1$ in $G$ and $\deg(w_i)=2$,
for every vertex $w_i\in W$.

Next, let $H$ be any $\maxdon$ realizing graph of $\sigma$.  Then $H$
must contain two vertices, say $x$ and $y$, of degree $d$, since a
single vertex of degree $d$ in $H$ can guarantee $\maxdon=d$ for at most
$d$ vertices. Next notice that $N[x]=N[y]$, because otherwise $H$ will
contain more than $d+1$ vertices. This implies that all the vertices
in $H$, other than $x$ and $y$, are adjacent to both $x$ and $y$.
Therefore, each of the vertices in $H$ must have degree at least two.
\end{proof}

 The next lemma shows that \addlayer outputs a valid list, even when the input list is
pseudo-valid.

\begin{lemma}
In procedure \addlayer, the list $\newL$ is valid even when
the list $\oldL$ is pseudo-valid and the parameter $d$ satisfies
$d\geq 3$.
\label{lemma:pseudo-valid}
\end{lemma}

\begin{proof}
We borrow notations from the proof of Lemma~\ref{lemma:valid}.  As
before, we have two separate cases depending on whether or not $k<d$.
We first consider the case $k\leq d-1$.  We showed in
Lemma~\ref{lemma:valid} that $(w_1,\cdots,w_k,a_1,\ldots,a_{i-1})$ is
a valid list of length at least $d-1$ when $\deg_{\oldH}(a_1)=1$.  We
now consider the scenario when $\oldL$ is pseudo-valid, and
$\deg_{\oldH}(a_1)=2$. The list $\newL$ is still valid if $k\geq 2$,
since the degree of $a_1$ in $\newH$ is at most $3$ and its position
in $\newL$ is also $3$ or greater. So the non-trivial case is $k=1$.
In such a case $i=d-1$, as the only vertex $w_1$ belonging to $W$ is
connected to $a_{d-1}$ in Algorithm~\ref{algo:Add-layer}.  Also,
$\deg_{\oldH}(a_{d-1})= 2$, and $a_{d-1}$ is connected to vertex
$w_1$, so to ensure that $\deg(a_{d-1}) = d$, in the for loop in step
9 of Algorithm~\ref{algo:Add-layer}, it is connected to only $d-3$
vertices, namely, $a_2,a_3,\ldots,a_{d-2}$. Since $a_{d-1}$ is never
connected to vertex $a_1$, $\deg_{\newH}(a_1) = \deg_{\oldH}(a_1)=2$.
This shows that the sequence $(w_1,\cdots,w_k,a_1,\ldots,a_{i-1}) =
(w_1,a_1,\ldots,a_{d-2})$ is a valid list of length exactly
$d-1$. Truncating it to length $d-2$ again yields a valid sequence.
In case $k\ge d$, $a_1$'s degree does not play any role, 
so the argument from the proof of Lemma~\ref{lemma:valid} works as is.
\end{proof}

\begin{remark}
The condition $d\geq 3$ is necessary in Lemma~\ref{lemma:pseudo-valid}
because in a pseudo-valid list all the vertices have degree $2$.
However, Procedure \addlayer works only in the case when the degree
of each vertex in the list is at most $d-1$, which does not hold
true for a pseudo-valid list when $d=2$.  
So we provide a different analysis for the profile $(d^{d+1},2^{k})$.
\end{remark}

\subsection{\maxdon realization of the profile $\sigma=(d^{d+1},2^k)$}
The following lemmas shows that 
$\sigma=(d^{d+1},2^1)$, for $d\geq 3$, is not $\maxdon$ realizable
when $d\geq 3$; and $\sigma=(d^{d+1},2^{k})$
is $\maxdon$ realizable when $d\geq 3$ and $k\geq 2$.

\begin{lemma}
For any integer $d\geq 3$, the profile $\sigma=(d^{d+1},2^1)$ is not $\maxdon$ realizable.
\label{lemma:d^d+1_2^1}
\end{lemma}
\begin{proof}
Let us assume on the contrary that $\sigma$ is $\maxdon$ realizable by 
a graph $G$, and let $w\in V(G)$ be a vertex such that $\maxdon(w)=2$.  
The graph $G$ must contain at least two vertices, say $x$ and $y$, of
degree $d$, since a single vertex of degree $d$ can guarantee $\maxdon$
of $d$ for at most $d$ vertices in the graph.  Consider the following
two cases.
\begin{compactenum}[(i)]
\item $N[x]=N[y]$: In this case the $\maxdon$ of all the vertices in
  $N[x]=N[y]$ is at least $d\geq 3$, as they are adjacent to either
  $x$ or $y$.  Thus $w\notin N[x]$, which implies that $V(G)=N[x]\cup
  \{w\}$ since $|N[x]|=d+1$ and $|V(G)|=d+2$.  Also, $w$ cannot be
  adjacent to any vertex in $N[x]$, because if $w$ is adjacent to a
  vertex $w_0 \in N[x]$, then $\deg(w_0)$ must be $3$, in
  contradiction to the assumption $\maxdon(w)=2$.  Thus the only
  possibility left is that $w$ is a singleton vertex, which is again a
  contradiction.

\item $N[x]\neq N[y]$: In this case the vertex set of $G$ is equal to
  $N[x]\cup N[y]$ since size of $N[x]\cup N[y]$ must be at least $d+2$
  (as $|N[x]\cap N[y]|\leq d$) and is also at most $|V(G)| = d+2$.
  This implies that all the vertices of $G$ are adjacent to either $x$
  or $y$, which contradicts the fact that $\maxdon(w)=2$, since
  $\deg(x)=\deg(y)=d\geq 3$.
\end{compactenum}
\vspace{-5mm}
\end{proof}

\begin{lemma}
For any integers $d\geq 3$ and $k\geq 2$, the profile
$\sigma=(d^{d+1},2^{k})$ is $\maxdon$ realizable.
Moreover, we can compute a connected realization in $O(d+k)$ time.
\label{lemma:d^d+1_2^k}
\end{lemma}
\begin{proof}
The construction of $G$ is as follows.  
Take a vertex $u_1$ and connect it to $d$ other vertices 
$v_1,\ldots,v_{d}$.  Next, take another vertex $u_2$ and connect it 
to vertices $v_2, \ldots, v_{d}$, and a new vertex $v_{d+1}$. 
Finally, take a path $(a_1,a_2,\ldots,a_{\alpha})$ on $\alpha=k-2$ 
new vertices, and connect $a_1$ to $v_{d+1}$.  
In the graph $G$, $\deg(u_1)=\deg(u_2) = d$, and $\deg(v_i), \deg(a_j)
\leq 2$, for $i \in [1,d+1]$ and $j\in [1,k-2]$.  Vertices $u_1$ and
$u_2$ has maximum degree in their neighborhood $2$, thus
$\maxdon(u_1)=\maxdon(u_2) = 2$. Each $v_i$ is adjacent to $u_1,u_2$, for
$i \in [1,d+1]$, so its $\maxdon$ is $d$. And, the $\maxdon$ of vertices
on the path $(a_1,a_2,\ldots,a_{\alpha})$ is $2$, since they have a
neighbour of degree~$2$.
\end{proof}

\subsection{Algorithm}
We now explain the construction of a graph realizing the profile
$\sigma=(d_\ell^{n_\ell}, \cdots, d_1^{n_1})\neq
(d_\ell^{d_\ell+1},2^1)$ that satisfies the conditions
\begin{inparaenum}[(i)]
\item $d_\ell \leq \min\{n_\ell,n-1\}$, and
\item $d_1\geq 2$
\end{inparaenum}, where $n=n_1+\cdots+n_\ell$.
If $\sigma$ is equal to $(d_\ell^{d_\ell+1},2^k)$, for some $k\geq 2$, we use
Lemma~\ref{lemma:d^d+1_2^k} to realize $\sigma$.  If not, then
depending upon the value of $n_\ell$, we initialize $G$ differently as
follows.  
(Refer to Algorithm~\ref{algo:maxNdegOpen} for the pseudocode).

\begin{algorithm}[!ht]
\caption{$\maxdon$ realization of $\sigma=(d_\ell^{n_\ell},\ldots, d_1^{n_1})$
\label{algo:maxNdegOpen}}
\Input{A sequence $\longsigma\neq ({d_\ell}^{d_\ell+1}2^1)$ satisfying $d_\ell\leq \min\{n-1,n_\ell\}$ and $d_1\geq 2$.}
\BlankLine
\BlankLine
\uIf {$\sigma=(d_\ell^{d_\ell+1},2^k)$ for some $k\geq 2$} 
{Use Lemma~\ref{lemma:d^d+1_2^k} to compute a realization $G$ for profile $\sigma$.\\}
\Else{
\uCase{$n_\ell\geq d_\ell+2$}{
Initialize $G$ to be the graph obtained from Lemma~\ref{lemma:simple-case-open}
that realizes the profile $(d_\ell^{n_\ell})$.\\
Set $L_{\ell-1}$ to be a valid list in $G$ of size $d_{\ell-1}-1$.\\
}
\uCase{$n_\ell=d_\ell+1$}{
Initialize $G$ to be the graph obtained from Lemma~\ref{lemma:d^d_dbar^1}
that realizes the profile $(d_\ell^{d_\ell+1})$.\\
Set $L_{\ell-1}$ to be a pseudo-valid list in $G$ of size $d_{\ell-1}-1$.\\
}
\Case{$n_\ell=d_\ell$}{
Initialize $G$ to be the graph obtained from Lemma~\ref{lemma:d^d+1}
that realizes the profile $(d_\ell^{d_\ell}d_{\ell-1})$.\\
Set $L_{\ell-1}$ to be a valid list in $G$ of size $d_{\ell-1}-1$.\\
Decrement $n_{\ell-1}$ by $1$.\\
}
\For{$(i=\ell-1$ to $1)$}{
$L_{i-1}\gets$ \addlayer$(G, L_{i},n_i, d_i)$.\\
Truncate list $L_{i-1}$ to contain only the first $d_{i-1}-1(\leq d_i -2)$
vertices.\\
}}
Output $G$.
\end{algorithm}

\begin{compactenum}
\item If $n_\ell\geq d_\ell+2$, we use Lemma~\ref{lemma:simple-case-open}
  to initialize $G$ to be a $\maxdon$ realization of the profile
  $(d_\ell^{n_\ell})$.  Recall $G$ contains an independent set, say
  $W=\{w_1,w_2,\ldots,w_{d_\ell}\}$, satisfying the condition that the
  degree of first two vertices is one, and the degree of the remaining
  vertices is at most two. We set $L_{\ell-1}$ to be the list
  $(w_1,w_2,\ldots,w_{d_{\ell-1}-1})$ (notice $d_{\ell-1}-1< d_\ell$).
  It is easy to verify that this list is valid.

\item If $n_\ell=d_\ell+1$, then a realization of
  $(d_\ell^{d_\ell+1})$ does not contains a valid list.  So we use
  Lemma~\ref{lemma:d^d+1} to initialize $G$ to be a
  $\maxdon$ realization of the profile $(d_\ell^{d_\ell+1})$ that
  contains a pseudo-valid list. This is possible since we showed $G$
  contains an independent set, say
  $W=\{w_1,w_2,\ldots,w_{d_\ell-1}\}$, such that degree of each $w\in
  W$ is two.  We set $L_{\ell-1}$ to be the list
  $(w_1,w_2,\ldots,w_{d_{\ell-1}-1})$ (again notice $d_{\ell-1}-1<
  d_\ell-1$).

\item If $n_\ell=d_\ell$, then the sequence $d_\ell^{d_\ell}$ is not
  realizable (see Lemma~\ref{lemma:suff_maxdDegOpen}).  So we initialize $G$
  to be the graph realization of $(d_\ell^{n_\ell},d_{\ell-1})$ as
  obtained from Lemma~\ref{lemma:d^d_dbar^1}.  We set $L_{\ell-1}$ be
  a valid list in $G$ of size $d_{\ell-1}-1$.  Also we decrement
  $n_{\ell-1}$ by one as $G$ already contain a vertex whose $\maxdon$ is
  $d_{\ell-1}$.
\end{compactenum}

Next for each $i=\ell-1$ to $1$ we perform following steps.
\begin{inparaenum}[(i)] 
\item We take as an input the valid list $L_{i}$ of size $d_{i}-1$,
  and execute Procedure \addlayer$(G, L_{i},n_i, d_i)$ to add $n_i$
  new vertices to $G$. The procedure returns a valid list $L_{i-1}$ of
  size $d_i-2$.
\item Truncate list $L_{i-1}$ to contain only the first
  $d_{i-1}-1(\leq d_i -2)$ vertices.  The truncated list
  remains valid since it is a prefix of a valid list.
\end{inparaenum}

\subparagraph{Correctness.}
Let $\bar V_\ell$ denote the set of vertices in graph $G$ initialized
in steps 5, 8, or 11 of Algorithm~\ref{algo:maxNdegOpen}, 
and for $i\in [1,\ell-1]$, let $\bar V_i$ denote the set of new vertices 
added to graph $G$ in iteration $i$ of for loop.
For $i\in [1,\ell]$, let $G_i$ be the graph induced by vertices
$\bar{V}_i \cup \cdots \cup \bar{V}_\ell$.

Recall that if $n_\ell=d_\ell$, then the graph is initialized in step
11 and contains $n_\ell+1$ vertices, of which one vertex, say~$z$, has
$\maxdon(z)=d_{\ell-1}$, and the remaining vertices have $\maxdon=d_\ell$.
If $n_\ell=d_\ell$, then let $Z=\{z\}$, otherwise let $Z=\emptyset$.
We set $V_\ell=\bar{V}_\ell \setminus Z$, $V_{\ell-1}=\bar{V}_{\ell-1}
\cup Z$, and $V_i = \bar{V}_i$ for $i\in [1,\ell-2]$.  Thus
$|V_i|=n_i$, for $i\in [1,\ell]$. The following lemma proves the correctness.


\begin{lemma}
\label{lemma:mdon-algo-correct}
For any $i\in [1,\ell]$, graph $G_i$ is a $\maxdon$ realization of
profile $(d_\ell^{n_\ell},\cdots,d_i^{n_i})$, except for the case 
$n_\ell=d_\ell$ in which $G_\ell$ is $\maxdon$ realization of profile $(d_\ell^{n_\ell},d_{\ell-1})$.
Moreover, for any $j\in [i,\ell]$, we have
\begin{compactenum}
\item For every $v\in V_j\setminus Z$, $\deg_{G_i}(v)\leq \maxdon_{G_i}(v) = d_j$.
\item If $n_\ell=d_\ell$, then $\deg_{G_i}(z)=d_\ell$ and
  $\maxdon_{G_i}(z) = d_{\ell-1}$.
\end{compactenum}
\end{lemma}

\begin{proof}
We prove the claim by induction on the iterations of the for loop.
The base case is for index $\ell$, and the claim follows from
Lemmas~\ref{lemma:simple-case-open},~\ref{lemma:d^d_dbar^1}, and
~\ref{lemma:d^d+1}.  Specifically, notice that every vertex $v \in
V_\ell$ that is included in $G$ in step 5, 8, or 11 of the algorithm
has $\maxdon(v) = d_\ell$.  In the case $n_\ell = d_\ell$, the vertex $
z \in V_{\ell-1}$ included in step 11 of algorithm has
$\maxdon(z) = d_{\ell-1}$.
Also, in both the cases, $V_\ell\cup Z$ is the vertex set of $G$, and degree
of all the vertices in this set is bounded by $d_\ell$.

For the inductive step, we assume that the claim holds for $i+1$, and
prove the claim for $i$.
Consider any vertex $v$ in $G_i$.  We have
two cases.
\begin{compactenum}
\item $v \in V_i\setminus Z~:$ In this case by
  Proposition~\ref{proposition:valid-open} and
  Lemma~\ref{lemma:pseudo-valid}, $\deg_{G_i}(v) \leq
  \maxdon_{G_i}(v) = d_i$.
\item $v \in V_j\setminus Z$, for $j > i~:$ In this case we first show 
  that for any vertex $w \in N(v)$, $\deg_{G_i}(w)\leq d_j$. 
  If $w \in V_i\setminus Z$, then we already showed 
  $\deg_{G_i}(w) \leq d_i$. So we next consider the case 
  $w\in (V_{i+1} \cup \cdots \cup V_{\ell})\setminus Z$.  
  Now if $w \in L_i$ participates in Procedure
  $\addlayer(G, L_{i},n_i, d_i)$, then by
  Proposition~\ref{proposition:valid-open} in the updated graph
  $\deg_{G_i}(w)\leq d_i \leq d_j$.  If $w\not\in L_i$, then 
  the degree of $w$ is unaltered in the $i^{th}$ iteration, and 
  thus $\deg_{G_i}(w)=\deg_{G_{i+1}}(w) \leq d_j$ by the 
  inductive hypothesis.  If $n_\ell=d_\ell$ and $w=z\in Z$, then 
  also $\deg_{G_i}(w)=\deg_{G_{i+1}}(w)$ since vertex $z$ 
  never participates in procedure $\addlayer$. It follows that 
  $\maxdon(v)$ remains unaltered due to iteration $i$, and thus 
  $\maxdon_{G_i}(v) = \maxdon_{G_{i+1}}(v) = d_j$.
\end{compactenum}

Now when $n_\ell=d_\ell$, then 
$\deg_{G_\ell}(z)=d_\ell$ and $\maxdon_{G_\ell}(z) = d_{\ell-1}$.
The degree of vertex~$z$ never changes since it does not participates
in procedure $\addlayer$. The $\maxdon$ of $z$ never changes
from the same reasoning as above.
\end{proof}


The execution time of algorithm takes $O\big(\sum_{i=1}^\ell (n_i+d_i)\big)$ time, which can be easily shown to be optimal. The following theorem is immediate from the above discussions.

\begin{theorem}
There exists an algorithm that given any profile
$\sigma=(d_\ell^{n_\ell},\ldots, d_1^{n_1})\neq ({d_\ell}^{d_\ell+1}2^1)$ 
with $n=n_1+\cdots +n_\ell$
satisfying $d_\ell\leq \min\{n-1,n_\ell\}$
and $d_1\geq 2$, computes in optimal time 
a connected $\maxdon$ realization of $\sigma$.
\label{theorem:mdon_algo}
\end{theorem}


\subsection{Complete characterization of $\maxdon$ profiles.}
We first give the sufficient conditions for a profile to be $\maxdon$ realizable.

\begin{lemma}
\label{lemma:suff_maxdDegOpen}
A necessary condition for the profile $\sigma=(d_\ell^{n_\ell}, \cdots, d_1^{n_1})$ with $n=n_1+\cdots+n_\ell$ to be $\maxdon$ realizable is $d_\ell \leq \min\{n_\ell, n-1\}$.
\end{lemma}
\begin{proof}
Suppose $\sigma$ is $\maxdon$ realizable by a graph $H$. Then there exists at least one vertex, say $u$, of degree exactly $d_\ell$ in $H$.  Now $|N(u)|=d_\ell$ and $|N[u]|=d_\ell+1$, which implies that the number of vertices in $H$ whose $\maxdon$ is $d_\ell$ must be at least $d_\ell$, so $n_\ell \geq d_\ell$. Also, the number of vertices in the graph $H$, $n$, must be at least $d_\ell+1$.
\end{proof}

Consider a profile $\sigma=(d_\ell^{n_\ell}, \cdots, d_1^{n_1})$
realizable by a connected graph.  If $d_1=1$, then the realizing graph
must contain a vertex, say $u$, such that each vertex in $N(u)$ has
degree $1$.  Let $d=\deg(u)$, and $v_1,\ldots,v_d$ be the neighbours of
$u$.  Then $\deg(v_1)=\cdots=\deg(v_d)=1$.  So in this case the
realizing graph is a star graph $K_{1,d}$ with $\maxdon$ profile
$\sigma=(d^d,1^1)$.  If $d_1\geq 2$, then by
Lemma~\ref{lemma:suff_maxdDegOpen}, for $\sigma$ to be realizable in this
case, we need that $d_\ell \leq \min\{n_\ell,n-1\}$.  Also,
Lemma~\ref{lemma:d^d+1_2^1} implies that $\sigma$ must not be
$(d^{d+1},2^1)$.  By Theorem~\ref{theorem:mdon_algo}, under these
conditions $\sigma$ is always realizable.  We thus have the following
theorem.

\begin{theorem} 
The necessary and sufficient condition  for a profile
$\sigma=(d_\ell^{n_\ell}, \cdots, d_1^{n_1}) \neq (d^{d+1},2^1)$
with $n=n_1+\cdots +n_\ell$
to be $\maxdon$ realizable by a connected graph is~
\begin{inparaenum}[(i)]
\item $d_\ell \leq \min\{n_\ell,n-1\}$ and $d_1\geq 2$; or 
\item $\sigma=(d^d,1^1)$ for some positive integer $d>1$; or
\item $\sigma=(1^2)$.
\end{inparaenum}
\label{theorem:mdon_connected}
\end{theorem}


For general graphs we have the following theorem.

\begin{theorem} 
The necessary and sufficient condition for a profile $\sigma$ to be $\maxdon$ realizable by a general graph is that $\sigma$ can be split\footnotetext{
A profile $\sigma=(d_\ell^{n_\ell},\cdots,d_1^{n_1})$ is said to be split into two profiles $\sigma_1=(d_\ell^{p_\ell},\cdots,d_1^{p_1})$ and $\sigma_2=(d_\ell^{q_\ell},\cdots,d_1^{q_1})$ if  $n_i=p_i+q_i$ for each $i\in[1,\ell]$.
}
into two profiles $\sigma_1$ and $\sigma_2$ such that
\begin{inparaenum}[(i)]
\item $\sigma_1$ has a connected $\maxdon$ realization, and 
\item $\sigma_2=(1^{2\alpha})$ or $\sigma_2=(d^d, 1^{2\alpha+1})$
  for some integers $d\geq 2, \alpha\geq 0$.
\end{inparaenum}
\label{theorem:mdon_general}
\end{theorem}

\begin{proof}
Suppose $\sigma$ is realizable by graph $G$. 
Let ${\cal C}(G)$ be a set consisting of all those components in $G$ 
that contain a vertex of $\maxdon$ equal to $1$ but is not an edge.
As a long as $|{\cal C}(G)|>1$, we perform following modifications to
$G$. Take any two components $C_1,C_2\in {\cal C}(G)$,
and let $\sigma_1$ and $\sigma_2$ be their $\maxdon$ profiles.
For $i=1,2$, component $C_i$ must be of form $K_{1,\delta_i}$ and contain 
$\delta_i(\geq 2)$ vertices of $\maxdon$ equal to $\delta_i$, and a 
single vertex of $\maxdon$ equal to 1. Let us assume $\delta_2\geq \delta_1$.
We replace $C_1$ and $C_2$ in $G$
by two different components, namely, an edge and 
(i) a connected $\maxdon$ realization of profile $\delta_2^{\delta_1+\delta_2}$ if $\delta_2=\delta_1$, or
(ii) a connected $\maxdon$ realization of profile $(\delta_2^{\delta_2},\delta_1^{\delta_1})$ if $\delta_2> \delta_1$.
In each iteration we decrease $|{\cal C}(G)|$ by a value two.
In the end if ${\cal C}(G)|$ is non-empty we denote the only component in it by $C_0$.
Next let $\bar C_1,\ldots,\bar C_k$ be all those components in $G$ that 
contain only the vertices of $\maxdon$ strictly greater than $1$.
Also let $\sigma_1,\ldots,\sigma_k$ be their $\maxdon$ profiles.
If $k>0$, we replace the components $\bar C_1,\ldots,\bar C_k$ by a single
connected component, say $\bar C_0$, that realizes the profile $\sigma_1+\cdots+\sigma_k$. 
It is easy to verify from Theorem~\ref{theorem:mdon_algo} that $\sigma_1+\cdots+\sigma_k$
will be $\maxdon$ realizable.
The final graph $G$ contains 
(i) at most one component, namely $\bar C_0$, having all vertices of $\maxdon$ greater than $1$,
(ii) at most one component, namely $C_0$, having exactly one vertex of $\maxdon$ equal to $1$, and
(iii) a union of some $\alpha\geq 0$ disjoint edges.
This shows that $\sigma$ can be split into two profiles $\sigma_1$ and $\sigma_2$ such that
\begin{inparaenum}[(i)]
\item $\sigma_1$ has a connected $\maxdon$ realization, and 
\item $\sigma_2=(1^{2\alpha})$ or $\sigma_2=(d^d, 1^{2\alpha+1})$
for some integers $d\geq 2, \alpha\geq 0$.
\end{inparaenum}
To prove the converse notice that $\sigma_2=(1^{2\alpha})$ is realizable by a disjoint union
of $\alpha\geq 0$ edges, and $\sigma_2=(d^d, 1^{2\alpha+1})$ is realizable
by a disjoint union of $\alpha$ edges and the star graph $K_{1,d}$.
Thus any $\sigma$ that can be split into two profiles $\sigma_1$ and $\sigma_2$
such that
\begin{inparaenum}[(i)]
\item $\sigma_1$ has a connected $\maxdon$ realization, and
\item $\sigma_2=(1^{2\alpha})$ or $\sigma_2=(d^d, 1^{2\alpha+1})$
for some integers $d\geq 2, \alpha\geq 0$
\end{inparaenum}
is $\maxdon$ realizable.
\end{proof}

\section{Concluding remarks on extremal neighborhood degree profiles}
\label{section:discussion}


Our work focuses on two similar neighborhood profiles, \maxdn and \mindn,
which capture two opposing extremes of the neighborhood,
but yet exhibit a surprising difference in structure.
The realizability of \maxdn profiles depends only on their prefix; in contrast,
the realizability characterization of \mindn profiles is incomplete
and depends on the entire profile.
Let us conclude with a brief discussion exploring the reasons behind this structural difference.

Let us first consider the \maxdn profile $\longsigma$  for a graph $G=(V,E)$. For $1\leq i\leq \ell$, let $W_i\subseteq V$ be the set of vertices whose
\maxdn in $G$ is at least $d_i$.
Note that for any vertex $v\in W_i$, a vertex having maximum degree in $N_G[v]$ (say $x$) must be contained in $W_i$. Moreover, all the neighbors of $x$ must also lie in $W_i$. It follows that the degree of $x$ remains unaltered when restricted to the induced subgraph $G[W_i]$, and $\maxdn_G(v)=\maxdn_{G[W_i]}(v)$. Hence, \maxdn profiles satisfy the following nice {\em substructure property}, which also justifies the incremental algorithm for computing their realizations given in Section~\ref{section:maxNdeg}:

\begin{sub-property}
The induced graph $G_i=G[W_i]$ is a $\maxdn$ realization  of the partial profile $(d_\ell^{n_\ell},\cdots,d_i^{n_i})$, for each $1\leq i\leq \ell$.
\vspace{-1.7mm}
\end{sub-property}

A natural question is whether a similar property holds for \mindn profiles. Unfortunately, in this case the answer is negative. To see why, consider the $\mindn$ profile $\longsigma$ for $G=(V,E)$. If $\mindn_G(v)$ is $d_i$ (for some $i$ and $v$), and $x=\arg\min\{\deg(x)~|~x\in N[v]\}$ is a leader of $v$, then the $\mindn$ of all vertices in $N[x]$ is {\em at most} $d_i$. But if we take $W_i$ to be the set of all vertices whose $\mindn$ in $G$ is {\em at most} $d_i$, and drop the vertices $z$ with $\mindn_G(z) > d_i$, i.e., look at the induced graph $G[W_i]$, then the degrees of $v$'s neighbors might {\em decrease}, so its leader might change. Hence the sub-structure property does not hold, which renders an incremental construction impossible, and contributes to the intricacy of realizing \mindn profiles.

Nevertheless, in this work we obtain a simple 2-approximate bound on the achievable $n_i$'s. The problem of obtaining an exact characterization for \mindn profiles is left as an interesting open question for future research.

%

\newpage
\appendix
\begin{center}
\huge{\bf Appendix~~}
\end{center}
\section{Counting the Number of Realizable MaxNDeg Sequences}
\label{section:counting}

We use the characterizations of the variants of \maxdn in order to count 
the number of realizable sequences.
Our results are summarized in the following theorem.

\begin{theorem}
\label{thm:counting}
For $n\ge 5$
\begin{itemize}
\item There are $2^{n-3}$ realizable sequences with connected graph in
  the closed neighborhood model.
\item There are $2^{n-2}-1$ realizable sequences with connected graph
  in the open neighborhood model.
\item There are $\ceil{\frac{2^{n-1}+(-1)^n}{3}}$ realizable sequences
  with any graph in the closed neighborhood model.
\item There are at least $\ceil{\frac{2^{n}-2}{3}} -
  \ceil{\frac{n-4}{2}}$ realizable sequences and at most $2^{n-1}-1$
  realizable sequences with any graph in the open neighborhood model.
\end{itemize}
\end{theorem}

There are $(n-1)^n$ unordered sequences $(d_n,\ldots,d_1)$ of length $n$
on the integers $1,\ldots,n-1$.
We count the number of non-increasing such sequences denoted by $S_n$.
Let $f(i,j,k)$ be the number of non-increasing sequences of length $k$
on the integers $i,\ldots,j$.  By definition, $S_n = f(1,n-1,n)$.

\begin{observation}
\label{obs:fijk}
$f(i,j,k) = \binom{k+j-i}{k}$.
\end{observation}
\begin{proof}
This is equivalent to counting the number of ways of placing $k$ balls
into $j-i+1$ ordered bins which is equivalent to inserting $j-i$
dividers among a line of $k$ balls.
\end{proof}

\begin{corollary}
\label{cor:Sn}
$S_n \approx \frac{4^{n-1}}{\sqrt{\pi n}}$.
\end{corollary}
\begin{proof}
$S_n = \binom{2n-2}{n}$ by Observation~\ref{obs:fijk}.
Stirling's formula implies the following approximation for the central
binomial coefficient $\binom{2n}{n} \approx \frac{4^n}{\sqrt{\pi n}}$.
The corollary follows since 
$\frac{\binom{2n}{n}}{\binom{2n-2}{n}} =
4+\frac{2}{n-1}$.
\end{proof}

Theorem~\ref{thm:counting} and Corollary~\ref{cor:Sn} imply that the
number of realizable sequences in all variants is roughly
$\Theta(\sqrt{S_n})$.


\subsection{Connected Graphs in the Closed Neighborhood Model}

Let $\ccon(n)$ be the number of length $n$ sequences that are \maxdn
realizable with a connected graph in the closed neighborhood model.
Recall that by Theorem~\ref{thm:char_closed_con} the sequence
$\sigma=(d_\ell^{n_\ell},\ldots,d_1^{n_1}) \in S_n$ can be realized
with a connected graph in the closed neighborhood model if and only if
one of the following holds:
\begin{inparaenum}[(i)]
\item $n = 2$: $\sigma=(1^2)$.
\item $n\ge 3$: $n_\ell \geq d_\ell-1$ and $d_1 \geq 2$.
\end{inparaenum}

\begin{lemma}
\label{lem:CCn}
$\ccon(2)=1$ and $\ccon(n) = 2^{n-3}$, for $ n \geq 3$.
\end{lemma}

\begin{proof}
By the first part of the characterization, $\ccon(2)=1$.
Assume $n \geq 3$.  Let $d=d_{\ell}$. By the second part of the
characterization, the first $d+1$ values in any realizable sequence
must be equal to $d$. The suffix of length $n-d-1$ is a non-increasing
sequence on the numbers $2,\ldots,d$. By the definition of $f(i,j,k)$
with $i=2$, $j=d$, and $k=n-d-1$ and by Observation~\ref{obs:fijk} the
number of such sequences is
$$
f(2,d,n-d-1)
= \binom{(n-d-1)+d-2}{n-d-1} 
= \binom{n-3}{d-2}
~.
$$
The value of $d$ ranges from $2$ to $n-1$. Hence, the total number of
realizable sequences is
$$
\ccon(n)
= \sum_{d=2}^{n-1} \binom{n-3}{d-2} 
= \sum_{i=0}^{n-3} \binom{n-3}{i} 
= 2^{n-3}
~.
$$
\end{proof}

\subsection{Connected Graphs in the Open Neighborhood Model}

Let $\ocon(n)$ be the number of length $n$ sequences that are \maxdn
realizable with a connected graph in the open neighborhood model.
Recall that by Theorem~\ref{theorem:mdon_connected} a sequence $\sigma
= (d_\ell^{n_\ell},\ldots,d_1^{n_1}) \in S_n$ can be realized with a
connected graph in the open neighborhood model if and only if one of
the following holds:
\begin{compactenum}[(i)]
\item $n=2$: $\sigma=(1^2)$.
\item $n\ge 3$: $\sigma=((n-1)^{n-1},1^1)$.
\item $n\ge 3$: $d_\ell \le n_\ell$, $d_1\ge 2$, and
  $\sigma\ne((n-2)^{n-1},2^1)$.
\end{compactenum}
Note that the sequence in item~2 is the only sequence in which one
vertex has a maximum degree 1 in its open neighborhood.  It is
realizable by the star graph.

\begin{lemma}
\label{lem:OCn}
$\ocon(2)=1$, $\ocon(3)=2$, $\ocon(4)=4$, and $\ocon(n)=2^{n-2}-1$,
for $n \ge 5$.
\end{lemma}
\begin{proof}
Following the characterization, one could verify the following: 
\begin{compactitem}
\item $(1,1)$ is the only realizable sequence of length
  $2$.  Therefore, $\ocon(2)=1$.
\item $(2,2,2)$ and $(2,2,1)$ are the only realizable sequences of
  length $3$.  Therefore, $\ocon(2)=2$.
\item $(3,3,3,3)$, $(3,3,3,2)$, $(3,3,3,1)$, and $(2,2,2,2)$ are the
  only realizable sequences of length $4$.  Therefore $\ocon(4)=4$.
\end{compactitem} 
Assume $n\ge 5$.  Let $d = d_\ell$.  For the sake of counting, we
assume that $((n-2)^{n-1},2^1)$ should be counted while
$((n-1)^{n-1},1^1)$ should not.  Hence, we need to count the sequences
for which $d_\ell \le n_\ell$ and $d_1 \ge 2$.
It follows that the number of realizable sequences with $d=d_{\ell}$
is the number of sequences in which the first $d$ values are equal to
$d$ and the suffix of length $n-d$ is a non-increasing sequence on the
numbers $2,\ldots,d$.  By Observation~\ref{obs:fijk} with $i=2$,
$j=d$, and $k=n-d$ the number of such sequences is
$$
f(2,d,n-d)
= \binom{(n-d)+d-2}{n-d} 
= \binom{n-2}{d-2}
~.
$$
The value of $d$ ranges from $2$ to $n-1$.  Hence, the total number of
realizable sequences is
$$
\ocon(n)
= \sum_{d=2}^{n-1} \binom{n-2}{d-2}
= \sum_{i=0}^{n-3} \binom{n-2}{i}
= 2^{n-2}-1
~.
$$
\end{proof}

Observe that $\ocon \approx 2 \cdot \ccon(n)$.  This is due to the
more relaxed constraint on $n_\ell$.

\subsection{General Graphs in the Closed Neighborhood Model}

Let $\cgen(n)$ be the number of length $n$ sequences that are \maxdn
realizable with a general graph in the closed neighborhood model.
%
%
By Theorem~\ref{thm:char_closed_gen} the sequence
$\sigma=(d_{\ell}^{n_{\ell}},\ldots,d_{1}^{n_{1}}) \in S_n$ can be
realized with a general graph (without isolated vertices) in the
closed neighborhood model if and only if the following holds for $n
\geq 2$: $d_\ell \le n_\ell-1$, and either $d_1 \ge 2$ or $n_1$ is
even.

\begin{lemma}
\label{lem:CG'n}
For $n \ge 2$, $\cgen(n) = (2^{n-1}+(-1)^n)/3$.
\end{lemma}
\begin{proof}
There are no realizable sequences of length $1$ and therefore
$\cgen(1)=0$.  The only realizable sequence of length $2$ is $(1,1)$
and therefore $\cgen(2)=1$.

Assume $n\ge 3$.  The first part of the characterization covers all
the realizations with connected graphs while the second part of the
characterization covers all the realizations with $n-2$ vertices plus
an isolated edge. As a result, we get the following recursive formula,
$$
\cgen(n)
= \cgen(n-2) + \ccon(n)
= \cgen(n-2) + 2^{n-3}
~.
$$
We prove by induction that the lemma follows from this recursion.  The
claim holds for the two base cases $n=1$ and $n=2$ since
$(2^0+(-1)^1)/3 = 0$ and $(2^1+(-1)^2)/3 = 1$.
Assume that the claim is correct for $n-2$, that is that $\cgen(n-2) =
(2^{n-3}+(-1)^{n-2})/3$.  It follows that $\cgen(n) =
(2^{n-3}+(-1)^{n-2})/3 + 2^{n-3} = (2^{n-1} + (-1)^n)/3$.
\end{proof}

\subsection{General Graphs in the Open Neighborhood Model}

Let $\ogen(n)$ be the number of length $n$ sequences that are \maxdn
realizable with a general graph in the open neighborhood model.

We do not know how to compute the exact value of $\ogen(n)$ based on
our complete characterization.  The main reason is that we do not know
how to avoid counting more than once a sequence that has several
realizations with one star graph where in each realization the size of
the star is different.  For example, consider the sequence
$(3^6,2^2,1^1)$. It can be realized with a $3$-regular graph of size
$6$ whose \maxdn sequence is $(3^6)$ and a star of size $3$ whose \maxdn
sequence is $(2^2,1^1)$.  It can also be realized by a cycle of size
$4$ that is connected to a vertex of degree $1$ whose \maxdn sequence is
$(3^3,2^2)$ and a star of size $4$ whose \maxdn sequence is $(3^3,1^1)$.
The problem is that the strategy of extracting the star and counting
the number of realizations for the remaining sequence would count more
than once sequences from which we can extract stars of different
sizes.

Instead we provide characterizations for under and over counting.  On
one hand, we count most of the sequences that can be realized and on
the other hand, we count all the realizable sequences, but also some
sequences that cannot be realized.  Specifically, we show two
functions $\ogenl(n)$ and $\ogenu(n)$ such that $\ogenl(n) \leq
\ogen(n) \leq \ogenu(n)$, for $n\ge 2$.

Let $\ogenl(n)$ be the number of sequences $\sigma =
(d_{\ell}^{n_{\ell}},\ldots,d_{1}^{n_{1}}) \in S_n$ that can be realized with
a general graph in the open neighborhood model if one of the following
holds for $n \geq 2$:
\begin{inparaenum}[(i)]
\item $d_{\ell} \le n_{\ell}$ and $d_1\ge 2$.
\item $d_{\ell} \le n_{\ell}$, $d_1=1$, and $n_1$ is even.
\end{inparaenum}

\begin{lemma}
\label{lemma:ogenl}
$\ogenl(n) \leq \ogen(n)$, for every $n \geq 1$.
\end{lemma}
\begin{proof}
$\ogenl(n)$ counts all the realizations with one connected component
  and a collection of isolated edges. The connected component could be
  a star. However, sequences that can be realized with a connected
  component, a star and a collection of isolated edges are not
  counted.
\end{proof}

\begin{lemma}
\label{lem:OGL}
$\ogenl(2) = 1$ and $\ogenl(n) = \ceil{(2^{n}-2)/3} - \ceil{(n-4)/2}$,
for $n\ge 3$.
\end{lemma}
\begin{proof}
One can verify the following:
\begin{compactenum}
\item The sequence $(1,1)$ is the only realizable sequence and
  therefore $\ogenl(2) = 1$.
\item The sequences $(2,2,2)$ and $(2,2,1)$ are realizable and
  therefore we can set $\ogenl(3) = 2$.
\item The sequences $(3,3,3,3)$, $(3,3,3,2)$, $(3,3,3,1)$,
  $(2,2,2,2)$, and $(1,1,1,1)$, are realizable and therefore we can
  set $\ogenl(4) = 5$.
\end{compactenum}
Lemma~\ref{lemma:ogenl} implies the following recessive formula for
$n\ge 4$,
$$
\ogen(n)
= \ogenl(n-2) + \ocon(n)
= \ogenl(n-2) + (2^{n-2}-1)
~.
$$
We prove by induction that the lemma follows from this recursion.
The claim holds for the two base cases $n = 3$ and $n = 4$ since
$\ceil{(2^{3}-2)/3} - \ceil{(3-4)/2} = 2$ and $\ceil{(2^{4}-2)/3} -
\ceil{(4-4)/2} = 5$.
The induction hypothesis for $n-2$ implies that $\ogenl(n) =
\ceil{(2^{n-2}-2)/3} - \ceil{(n-6)/2} + (2^{n-2}-1)$.  For an even
$n$, we have
$$
\ogenl(n)
= \frac{2^{n-2}-1}{3} - \frac{n-6}{2} + (2^{n-2}-1) 
= \frac{2^{n}-1}{3} - \frac{n-4}{2} 
= \ceil{\frac{2^{n}-2}{3}} - \ceil{\frac{n-4}{2}}
~,
$$
and for an odd $n$, 
$$
\ogenl(n)
= \frac{2^{n-2}-2}{3} - \frac{n-5}{2} + (2^{n-2}-1) 
= \frac{2^{n}-2}{3} - \frac{n-3}{2} 
= \ceil{\frac{2^{n}-2}{3}} - \ceil{\frac{n-4}{2}}
~.
$$
\end{proof}

Let $\ogenu(n)$ be the number of non-increasing sequences $\sigma =
(d_{\ell}^{n_{\ell}},\ldots,d_{1}^{n_{1}}) \in S_n$ that satisfy
$d_\ell \le n_\ell$ for $n\ge 2$.

\begin{lemma}
$\ogenu(n) \geq \ogen(n)$, for every $n \geq 1$.
\end{lemma}
\begin{proof}
By Theorem~\ref{theorem:mdon_general} In any realizable sequence,
$d_{\ell}$ cannot be larger than $\min\{n_{\ell},n-1\}$.
\end{proof}

\begin{lemma}
\label{lem:OGU}
$\ogenu(2) = 1$ and $\ogenu(n) = 2^{n-1}-1$, for $n\ge 2$.
\end{lemma}
\begin{proof}
For $n = 2$, $(1,1)$ is the only sequence and therefore $\ogenu(2) =
1$.
Assume $n\ge 2$.
Let $d = d_{\ell}$.  The first $d$ values in any realizable sequence
must be equal to $d$.  The suffix of length $n-d$ is a non-increasing
sequence on the numbers $1,\ldots,d$.  By Observation~\ref{obs:fijk}
with $i=1$, $j=d$, and $k=n-d$ the number of such sequences is
$$
f(1,d,n-d)
= \binom{(n-d)+d-1}{n-d} 
= \binom{n-1}{d-1}
~.
$$
The value of $d$ ranges from $1$ to $n-1$.  Hence, the total number of
realizable sequences is
$$
\ogenu(n)
= \sum_{d=1}^{n-1} \binom{n-1}{d-1} 
= \sum_{i=0}^{n-2} \binom{n-1}{i} 
= 2^{n-1}-1
~.
$$
\end{proof}

Observe that the ratio between the upper bound and the lower bound is about $3/2$.  



\begin{thebibliography}{10}

\bibitem{AT94}
Martin Aigner and Eberhard Triesch.
\newblock Realizability and uniqueness in graphs.
\newblock {\em Discrete Mathematics}, 136:3--20, 1994.

\bibitem{BCPR18survey}
Amotz Bar-Noy, Keerti Choudhary, David Peleg, and Dror Rawitz.
\newblock Realizability of graph specifications: Characterizations and
  algorithms.
\newblock In {\em 25th SIROCCO}, volume 11085 of {\em LNCS}, pages 3--13, 2018.

\bibitem{BCPR19-range-isaac}
Amotz Bar{-}Noy, Keerti Choudhary, David Peleg, and Dror Rawitz.
\newblock Efficiently realizing interval sequences.
\newblock In {\em 30th International Symposium on Algorithms and Computation,
  {ISAAC} 2019, December 8-11, 2019, Shanghai University of Finance and
  Economics, Shanghai, China}, pages 47:1--47:15, 2019.

\bibitem{BCPR19happiness}
Amotz Bar-Noy, Keerti Choudhary, David Peleg, and Dror Rawitz.
\newblock Graph profile realizations and applications to social networks.
\newblock In {\em 13th WALCOM}, volume 11355 of {\em LNCS}, pages 1--12, 2019.

\bibitem{BarrusD:18}
Michael~D. Barrus and Elizabeth~A. Donovan.
\newblock Neighborhood degree lists of graphs.
\newblock {\em Discrete Mathematics}, 341(1):175--183, 2018.

\bibitem{BD11}
Joseph~K. Blitzstein and Persi Diaconis.
\newblock A sequential importance sampling algorithm for generating random
  graphs with prescribed degrees.
\newblock {\em Internet Mathematics}, 6(4):489--522, 2011.

\bibitem{choudum86}
Sheshayya~A. Choudum.
\newblock A simple proof of the {E}rd\"os-{G}allai theorem on graph sequences.
\newblock {\em Bulletin of the Australian Mathematical Society}, 33(1):67--70,
  1991.

\bibitem{Cloteaux16}
Brian Cloteaux.
\newblock Fast sequential creation of random realizations of degree sequences.
\newblock {\em Internet Mathematics}, 12(3):205--219, 2016.

\bibitem{EG60}
Paul Erd\"os and Tibor Gallai.
\newblock Graphs with prescribed degrees of vertices [hungarian].
\newblock {\em Matematikai Lapok}, 11:264--274, 1960.

\bibitem{hakimi62}
S.~Louis Hakimi.
\newblock On realizability of a set of integers as degrees of the vertices of a
  linear graph --{I}.
\newblock {\em SIAM J. Appl. Math.}, 10(3):496--506, 1962.

\bibitem{havel55}
V.~Havel.
\newblock A remark on the existence of finite graphs [in {C}zech].
\newblock {\em Casopis Pest. Mat.}, 80:477--480, 1955.

\bibitem{K57}
P.J. Kelly.
\newblock A congruence theorem for trees.
\newblock {\em Pacific J. Math.}, 7:961--968, 1957.

\bibitem{MV02}
Milena Mihail and Nisheeth Vishnoi.
\newblock On generating graphs with prescribed degree sequences for complex
  network modeling applications.
\newblock {\em 3rd Workshop on Approximation and Randomization Algorithms in
  Communication Networks}, 2002.

\bibitem{MR15}
Elchanan Mossel and Nathan Ross.
\newblock Shotgun assembly of labeled graphs.
\newblock {\em CoRR}, abs/1504.07682, 2015.

\bibitem{O70}
Peter~V. O'Neil.
\newblock Ulam's conjecture and graph reconstructions.
\newblock {\em Amer. Math. Monthly}, 77:35--43, 1970.

\bibitem{SH91}
Gerard Sierksma and Han Hoogeveen.
\newblock Seven criteria for integer sequences being graphic.
\newblock {\em Journal of Graph Theory}, 15(2):223--231, 1991.

\bibitem{TT08}
Amitabha Tripathi and Himanshu Tyagi.
\newblock A simple criterion on degree sequences of graphs.
\newblock {\em Discrete Applied Mathematics}, 156(18):3513--3517, 2008.

\bibitem{U60}
S.M. Ulam.
\newblock {\em A collection of mathematical problems}.
\newblock Wiley, 1960.

\bibitem{WK73}
D.L. Wang and D.J. Kleitman.
\newblock On the existence of $n$-connected graphs with prescribed degrees
  ($n>2$).
\newblock {\em Networks}, 3:225--239, 1973.

\bibitem{W99}
N.C. Wormald.
\newblock Models of random regular graphs.
\newblock {\em Surveys in combinatorics}, 267:239--298, 1999.

\end{thebibliography}
\end{document}